\documentclass[11pt,letterpaper]{article}

\usepackage[margin=1in]{geometry}
\usepackage[english]{babel}
\usepackage{xspace}
\usepackage{amssymb}
\usepackage{amsmath}
\usepackage{amsthm}
\usepackage{graphicx}
\usepackage{multicol}
\usepackage{diagbox}
\usepackage[colorlinks=true, allcolors=teal]{hyperref}
\usepackage{enumerate}
\usepackage{algpseudocode}
\usepackage{tikz}
\usepackage{pgfplots}
\pgfplotsset{compat=1.17}
\usetikzlibrary{patterns}
\usepackage{pgfplotstable}
    \pgfplotsset{
    name nodes near coords/.style={
        every node near coord/.append style={
            name=#1-\coordindex,
            alias=#1-last,
        },
        every axis/.append style={
        /pgf/number format/.cd,
        precision=3,
        }, 
    },
    name nodes near coords/.default=coordnode
    }
\usepgfplotslibrary{fillbetween}
\usepackage{subfigure}
    
\usepackage[ruled,linesnumbered,vlined]{algorithm2e}

\newtheorem{theorem}{Theorem}[section]
\newtheorem{lemma}[theorem]{Lemma}
\newtheorem{corollary}[theorem]{Corollary}
\newtheorem{claim}[theorem]{Claim}
\newtheorem{definition}[theorem]{Definition}

\newenvironment{proofof}[1]{{\vspace*{5pt} \noindent\bf Proof of #1:  }}{\hfill\rule{2mm}{2mm}\vspace*{5pt}}

\newcommand{\opt}{\mathsf{OPT}}
\newcommand{\alg}{\mathsf{ALG}}

\newcommand{\Ranking}{\textsc{Ranking}\xspace}
\newcommand{\E}{\mathbb{E}}

\usepackage{color, xcolor} 
\usepackage{framed}
\definecolor{shadecolor}{named}{lightgray}

\DeclareMathOperator*{\argmin}{argmin}

\tikzstyle{state}=[
    circle,
    minimum size = 1.25pt
    draw = black,
    thick,
    text = black
]

\title{Degree-bounded Online Bipartite Matching: OCS vs. Ranking}

\author{Yilong Feng
	\thanks{State Key Lab of IOTSC, University of Macau. \{yc37459,yc47435,xiaoweiwu,yc17423\}@um.edu.mo. The authors are ordered alphabetically.}
	\and Haolong Li $^*$
        \and Xiaowei Wu $^*$
	\and Shengwei Zhou $^*$}

\date{}

\begin{document}

\begin{titlepage}
    
\thispagestyle{empty}
\maketitle

\begin{abstract}
We revisit the online bipartite matching problem on $d$-regular graphs, for which Cohen and Wajc (SODA 2018) proposed an algorithm with a competitive ratio of $1-2\sqrt{H_d/d} = 1-O(\sqrt{(\log d)/d})$ and showed that it is asymptotically near-optimal for $d=\omega(1)$.
However, their ratio is meaningful only for sufficiently large $d$, e.g., the ratio is less than $1-1/e$ when $d\leq 168$.
In this work, we study the problem on $(d,d)$-bounded graphs (a slightly more general class of graphs than $d$-regular) and consider two classic algorithms for online matching problems: \Ranking and Online Correlated Selection (OCS).
We show that for every fixed $d\geq 2$, the competitive ratio of OCS is at least $0.835$ and always higher than that of \Ranking.
When $d\to \infty$, we show that OCS is at least $0.897$-competitive while \Ranking is at most $0.816$-competitive.
We also show some extensions of our results to $(k,d)$-bounded graphs.
\end{abstract}

\end{titlepage}

\section{Introduction}

As one of the central problems in the field of computer science, the online bipartite matching problem has attracted enormous attention since the work of Karp, Vazirani, and Vazirani~\cite{conf/stoc/KarpVV90}.
The interests in this problem are mostly motivated by internet advertising applications, which contribute to a business with over \$259 billion in 2024 in the United States alone~\cite{IAB-Report-24}.
The problem can be modeled as an online matching problem on a bipartite graph $G(S\cup R, E)$, where the offline vertices $S$ (which we call the \emph{servers}) represent the advertisers and the online vertices $R$ (which we call the \emph{requests}) correspond to the ad slots of webpages.
Upon the arrival of an online request $r\in R$, its set of neighbors in $S$ is revealed and the online algorithm must decide irrevocably which neighbor the request matches, subject to the constraint that each server can be matched at most once.
The online algorithm is measured by its \emph{competitive ratio}, which is the worst ratio over all online instances, between the (expected) size of matching computed by the algorithm, and that of the offline maximum matching.
In their seminal work, Karp et al.~\cite{conf/stoc/KarpVV90} proposed the \Ranking algorithm that achieves a competitive ratio of $1-1/e \approx 0.632$, and showed that this is the best possible ratio for any online algorithms.
However, the theoretical results, while being optimal, provide a rather low guarantee on the competitive ratio, which is below expectation for an industry involving hundreds of billions of dollars.
Therefore, researchers have proposed several weaker models for which better competitive ratios are possible, including the random arrival model~\cite{conf/soda/GoelM08,conf/stoc/KarandeMT11, conf/stoc/MahdianY11,journals/talg/HuangTWZ19} and the stochastic arrival (known i.i.d.) model~\cite{conf/focs/FeldmanMMM09,journals/mor/ManshadiGS12,conf/stoc/HuangS21,conf/soda/Yan24}.
Unlike these models that restrict the arrivals of requests, Naor and Wajc~\cite{journals/teco/NaorW18} proposed a model that puts degree constraints on the bipartite graph.

\paragraph{Degree-bounded Graphs.}
Naor and Wajc~\cite{journals/teco/NaorW18} proposed to study the online bipartite matching problem on $(k,d)$-bounded graphs, where the degree of each server is at least $k$ and the degree of each request is at most $d$, with $k\geq d\geq 2$\footnote{Note that any greedy algorithm is $1$-competitive if $d=1$.}.
The motivation of the proposed model is that in real-world scenarios, servers are often interested in a large number of requests while requests are interesting to a small number of servers.
Below we present two specific examples:

\begin{shaded}
\noindent
\textbf{Application 1 (Job Market).} 
Consider a job market where the positions can be treated as offline servers while the job hunters arrive online as requests.
When a candidate looks for jobs in an employment platform, the platform needs to match her to potential positions of interest.
The number of positions a user can be matched with is naturally limited (e.g., dozens) due to the industry, educational background, and time constraints.
On the other hand, in real life, each job position is often interesting to many candidates (e.g., hundreds or thousands).
\end{shaded}

\begin{shaded}
\noindent
\textbf{Application 2 (Online Ads).} 
As another example, consider the case of targeted advertising where ad slots are denoted by online requests and advertisers are denoted by offline servers.
Given that advertisers generally tailor their advertising efforts towards particular demographic groups~\cite{journals/teco/NaorW18}, the number of rivals an advertiser encounters (for a particular ad slot) is typically restricted. 
For example, an ad slot on a webpage related to smartphones is interesting only to the advertisers that produce IT hardware, which are not too many.
On the other hand, due to the vast amount of online ad slots, each server is often interested in many online ad slots.
\end{shaded}

For the online bipartite matching problem on $(k,d)$-bounded graphs, Naor and Wajc~\cite{journals/teco/NaorW18} proposed a deterministic algorithm called \textsc{High-Degree} that achieves a competitive ratio of $1-(1-1/d)^{k}$.
Moreover, they show that the competitive ratio is optimal for all deterministic algorithms.
They have also considered randomized algorithms, and show that \textsc{Random} (every online request proposes to a randomly selected neighboring server, regardless of its matching status) achieves the same competitive ratio of $1-(1-1/d)^{k}$.
As one of the open questions, they propose to study randomized algorithms that achieve competitive ratios strictly higher than $1-(1-1/d)^{k}$.

The open problem is partially answered by Cohen and Wajc~\cite{conf/soda/CohenW18}, who studied the problem on $d$-regular graphs (a special case of $(d,d)$-bounded graphs), and proposed the \textsc{Marking} algorithm that achieves a competitive ratio of $1-2\sqrt{H_d/d} = 1-O(\sqrt{(\log d)/d}))$.
Their result can also be extended to vertex-weighted graphs (in which each offline server has a non-negative weight)~\cite{conf/soda/CohenW18} and $(k,d)$-bounded graphs~\cite{phd/us/Wajc21}, with slightly worse competitive ratios.
They also proved that no algorithm can achieve a competitive ratio of $1-o(\sqrt{1/d})$, showing that their competitive ratio is asymptotically near-optimal.
However, observe that the competitive ratio $1-2\sqrt{H_d/d}$ is meaningful only when $d$ is sufficiently large\footnote{Their algorithm \textsc{Marking} is designed specifically for large $d$. 
For small $d$, the performance of \textsc{Marking} could be sub-optimal (see a detailed discussion in Appendix~\ref{sec:hardness-for-marking}).}:
for $d\leq 79$, the ratio is smaller than $0.5$;
for $d\leq 168$, the ratio is smaller than $1-1/e$.
In other words, for the case when $d = O(1)$, the problem is largely unexplored.
On the other hand, the aforementioned applications suggest that the domain of small $d$ is the most central case of degree-bounded graphs.

Over the past few years, the rapid development of online matching algorithms has led to significant successes, particularly in the application and analysis of classic algorithms. Notably, the \textsc{Ranking} algorithm and the Online Correlated Selection (OCS) algorithm have demonstrated remarkable competitiveness. However, despite their significant success, very little is known regarding their performance in degree-bounded settings. An exception is the work of Cohen and Wajc~\cite{conf/soda/CohenW18}, who established an upper bound of $0.875+o(1)$ for \Ranking on $d$-regular graphs for $d=\omega(1)$ but had very few discussion on the case of general $d$.

To address the above research questions, we revisit the online bipartite matching problem on $d$-regular graphs\footnote{All of our results hold for $(d,d)$-bounded graphs; some of them can be extended to $(k,d)$-bounded graphs. For convenience of discussion, we only discuss our results on $d$-regular graphs in the introduction.}, where $d$ can be any positive integer.
We consider two of the most classic algorithms for online matching: \Ranking and OCS, and bound their competitive ratios.

\paragraph{Ranking.}
The \Ranking algorithm was first proposed by Karp et al.~\cite{conf/stoc/KarpVV90} for the (unweighted) online bipartite matching problem and shown to achieve the optimal competitive ratio of $1-1/e$.
The proof was simplified by a series of follow-up works~\cite{journals/sigact/BirnbaumM08,conf/soda/GoelM08,conf/soda/DevanurJK13,conf/sosa/EdenFFS21} and the algorithm has been extended to the vertex-weighted setting~\cite{conf/soda/AggarwalGKM11,conf/soda/DevanurJK13} and the random arrival setting~\cite{conf/stoc/MahdianY11,conf/stoc/KarandeMT11,conf/wine/JinW21}.
Besides the great success in online bipartite matching, the algorithm has also found applications to other matching problems, including oblivious matching~\cite{journals/siamcomp/ChanCWZ18,journals/jacm/TangWZ23}, fully online matching~\cite{journals/jacm/HuangKTWZZ20,conf/focs/0002T0020}, online matching with stochastic rewards~\cite{conf/focs/MehtaP12,conf/wine/HuangJSSWZ23} and Adwords (with unknown budgets)~\cite{conf/sigecom/Udwani23}.

\paragraph{Online Correlated Selection.}
The Online Correlated Selection (OCS) technique was first introduced by Fahrbach et al.~\cite{journals/jacm/FahrbachHTZ22} for the edge-weighted online bipartite matching problem and was shown to achieve a competitive ratio of $0.5086$ -- the first to surpass $0.5$.
The ratio was then improved by subsequent works~\cite{conf/isaac/ShinA21,conf/focs/GaoHHNYZ21}, resulting in the state-of-the-art competitive ratio of $0.5368$ by Blanc and Charikar~\cite{conf/focs/BlancC21}.
As a general and powerful framework to round fractional solutions online, the OCS-based algorithms have been applied to many other matching problems, including online stochastic matching~\cite{conf/stoc/TangWW22,conf/stoc/0002SY22}, Adwords (with general bids)~\cite{journals/siamcomp/HuangZZ24} and fair matching~\cite{conf/aaai/Hosseini00023}.

\subsection{Our Contribution}

We study the performance of \Ranking and OCS-based algorithms\footnote{Precisely speaking, the OCS-based algorithm we consider is based on the multi-way semi-OCS in \cite{conf/focs/GaoHHNYZ21}.} for online bipartite matching on $(d,d)$-bounded graphs.
In contrast to existing works, in which \Ranking often outperforms OCS, we show that on $(d,d)$-bounded graphs, the competitive ratio of OCS is strictly larger than that of \Ranking for every fixed $d\geq 2$.

\smallskip

Our first result shows a sequence of $d$-regular hard instances for which we give upper bounds on the competitive ratio of \Ranking (for every $d\geq 2$).
Specifically, our upper bound is
$119/144\approx 0.8264$ when $d = 2$ and
$1-1/(2e) \approx 0.8161$ when $d\to \infty$, improving the previous upper bound of $0.875+o(1)$ by~\cite{conf/soda/CohenW18}.
We remark that while the hard instances for \Ranking are often very well structured, theoretically establishing an upper bound on its competitive ratio turns out to be technically difficult in existing works.
Many of these upper bounds are achieved by upper bounding the competitive ratio of \textsc{Balance} (a.k.a. \textsc{Water-Filling})~\cite{conf/stoc/KarpVV90,journals/jacm/HuangKTWZZ20}, since the performance of \Ranking on a graph is equivalent to that of \textsc{Balance} if every vertex is replaced by infinitely many identical copies\footnote{In our case, such a reduction does not work due to the existence of the degree constraints.}.
Otherwise, the upper bounds are either established by formulating and solving complex differential equations~\cite{conf/stoc/KarandeMT11,journals/jacm/HuangKTWZZ20,conf/esa/LiangTX0Z23}\footnote{In fact, some details are missing in the proof of~\cite[Theorem 1.4]{journals/jacm/HuangKTWZZ20} for establishing the upper bound on the competitive ratio of \Ranking. The analysis we present in this paper can be seen as a complement to their proof.}, or suggested by simulations~\cite{journals/siamcomp/ChanCWZ18}.
In this work, we show the theoretical upper bound on the competitive ratio of \Ranking on the hard instances via an interesting Markov chain interpretation of the random process, and using Jensen's Inequality to bound the expected performance.
We believe that our analysis framework will be of independent interest to future works that study upper bounds for \Ranking.

\smallskip

Our second result establishes a lower bound of at least $0.835$ for the competitive ratio of the OCS-based algorithm (for every $d\geq 2$).
%
Observe that the ($d$-way semi-) OCS algorithm of~\cite{conf/focs/GaoHHNYZ21} can be directly applied to the problem on $(d,d)$-bounded graphs, giving a competitive ratio of $0.875$ when $d=2$ (optimal) and $1-e^{-1.679} \approx 0.813$ when $d\geq 3$.
In this paper, we show that the competitive ratio can be improved, for every $d\geq 3$.
We follow the analysis framework of~\cite{conf/focs/GaoHHNYZ21} and reduce the problem of lower bounding the competitive ratio of OCS to an optimization problem of designing a candidate function meeting certain constraints.
Interestingly, we show that for $d$-regular graphs, this optimization problem can be solved \emph{optimally}.
We show that this optimal solution gives a competitive ratio of at least $0.8352$ when $d=3$ and $0.8976$ when $d\to \infty$.
Moreover, since our performance guarantee for the algorithm is per-vertex, our results can be extended to \emph{vertex-weighted} graphs (where offline servers have weights).

As a consequence, we separate the performance of OCS and \Ranking on $d$-regular graphs, showing that OCS is strictly better than \Ranking on $(d,d)$-bounded graphs, for every fixed $d\geq 2$.
We summarize our results in Figure~\ref{fig:results}.
We also show (in Appendix~\ref{appendix:kd}) that some of our results can be extended to $(k,d)$-bounded graphs, for every $k\geq d\geq 2$. 

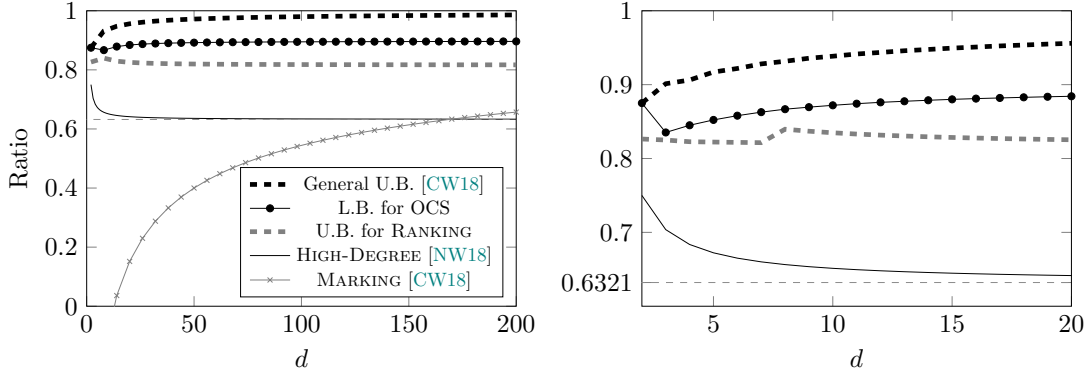
\begin{figure}[htbp]
\centering
\resizebox{0.9\textwidth}{!}{
\begin{subfigure}
    \centering
    \begin{tikzpicture}
        \begin{axis}[
            ymax=1,
            ymin=0,
            enlarge x limits=false,
            xmax=200,
            xmin=0,
            xlabel = {$d$},
            ylabel = {Ratio},
            height=6cm,
            width=8cm,
            legend pos=south east,
            legend style={nodes={scale=0.75, transform shape}}
        ]
            \addplot[black, no marks, dashed, line width=2pt] table [x=d, y=UB, col sep=comma, each nth point=6] {General_UB.csv};    
            \addlegendentry{General U.B.~\cite{conf/soda/CohenW18}}
            \addplot[mark=*, mark size=1.5, draw=black] table [x=d, y=OCS, col sep=comma, each nth point=6] {General_UB.csv};     
            \addlegendentry{L.B. for OCS}
            \addplot[gray, no marks, dashed, line width=2pt] table [x=d, y=RANKING, col sep=comma, each nth point=6] {General_UB.csv};    
            \addlegendentry{U.B. for \Ranking}
            \addplot[no marks, mark size=1.5, draw=black] table [x=d, y=DETERMINISTIC, col sep=comma] {General_UB.csv};    
            \addlegendentry{\textsc{High-Degree} \cite{journals/teco/NaorW18}}
            \addplot[mark=x, mark size=1.5, draw=gray] table [x=d, y=SODA, col sep=comma, each nth point=6] {General_UB.csv};    
            \addlegendentry{\textsc{Marking} \cite{conf/soda/CohenW18}}
            \addplot[domain=-10:210, dashed, draw=gray!90]{0.6321};
        \end{axis}
    \end{tikzpicture}
\end{subfigure}%
\begin{subfigure}
    \centering
    \begin{tikzpicture}
        \begin{axis}[
            ymax=1,
            ymin=0.6,
            enlarge x limits=false,
            xmax=20,
            xmin=2,
            xlabel = {$d$},
            ytick={0.6321,0.7,0.8,0.9,1},
            yticklabels={0.6321,0.7,0.8,0.9,1},
            height=6cm,
            width=8cm,
        ]
            \addplot[mark=*, mark size=1.5, draw=black] table [x=d, y=OCS, col sep=comma] {General_UB.csv};     
            %
            \addplot[black, no marks, dashed, line width=2pt, name path = A] table [x=d, y=UB, col sep=comma] {General_UB.csv};    
            \addplot[gray, no marks, dashed, line width=2pt] table [x=d, y=RANKING, col sep=comma] {General_UB.csv};    
            \addplot[no marks, mark size=1.5, draw=black] table [x=d, y=DETERMINISTIC, col sep=comma] {General_UB.csv};    
            \addplot[domain=-10:210, dashed, draw=gray!90]{0.6321};
    
    
        \end{axis}
\end{tikzpicture}
\end{subfigure}
}
\vspace{-10pt}
\caption{
The comparison of competitive ratios of OCS, \Ranking, \textsc{High-Degree}~\cite{journals/teco/NaorW18} and \textsc{Marking}~\cite{conf/soda/CohenW18} when $d \leq 200$ and $d\leq 20$, respectively, where L.B. and U.B. stand for lower bound and upper bound.
The general upper bound is from~\cite{conf/soda/CohenW18}, who analyzed a hard instance when $d\to \infty$. 
In Appendix~\ref{sec:problem-hardness} we refine their analysis to give an upper bound for every $d\geq 2$.}
\label{fig:results}
\end{figure}

\paragraph{OCS vs. Ranking.}
On a high level, both \Ranking and OCS can be regarded as online rounding algorithms for fractional matchings, and \Ranking often outperforms OCS in existing works.
For example, for the unweighted and vertex-weighted online bipartite matching, \Ranking achieves the optimal ratio of $1-1/e\approx 0.632$, while the current best ratio of OCS-based algorithms is $0.593$~\cite{conf/focs/GaoHHNYZ21}.
For online bipartite matching with reusable resources where each server becomes available again after a deterministic period of time of being matched, Delong et al.~\cite{delong2023online} showed that a \Ranking-based algorithm achieves the state-of-the-art competitive ratio of $0.589$ while the OCS-based algorithm is $0.505$-competitive.
However, we observe that for $d$-regular graphs, the number of past appearances of a server (i.e., the current degree, which is crucial in the OCS-based algorithms) might be more useful information to utilize than the artificial ranks used by \Ranking.
Consider the following toy example (Figure~\ref{fig:OCS-ranking-on-simple-graph}).
Every reasonable randomized algorithm (including OCS and \Ranking) will match $s_1$ and $s_2$ with equal probability when request $r_1$ arrives.
However, the behaviors of (the optimal $2$-way semi-) OCS and \Ranking become very different when $r_2$ arrives, if both $s_2$ and $s_3$ are available.
Since $s_2$ has appeared once as a neighbor of $r_1$ but was not chosen, OCS will match $r_2$ with $s_2$ with probability $1$, while conditioned on the rank of $s_2$ being larger than that of $s_1$, the probability that $r_2$ matches $s_2$ is only $1/3$ in \Ranking.

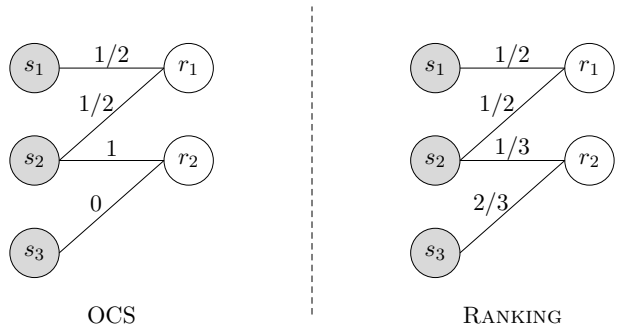
\begin{figure}[htb]
	\begin{center}
		\resizebox{0.5\textwidth}{!}{
			\begin{tikzpicture}
				\draw [fill = gray!30] (2.5,0) circle (0.4); \node at (2.5,0) {$s_3$};
				\draw [fill = gray!30] (2.5,1.5) circle (0.4); \node at (2.5,1.5) {$s_2$};
				\draw [fill = gray!30] (2.5,3) circle (0.4); \node at (2.5,3) {$s_1$};
				\draw (5,1.5) circle (0.4); \node at (5,1.5) {$r_2$};
				\draw (5,3) circle (0.4); \node at (5,3) {$r_1$};
				\draw (2.9,3)--(4.6,3); 
				\node at (3.75,3.2) {$1/2$};
				\draw (2.9,1.5)--(4.6,3); 
				\node at (3.5,2.4) {$1/2$};
				\draw (2.9,1.5)--(4.6,1.5); 
				\node at (3.75,1.7) {$1$};
				\draw (2.9,0)--(4.6,1.5); 
				\node at (3.5,0.8) {$0$};
				\node at (3.75,-1) {OCS};
				\draw [densely dashed] (7,-1)--(7,4); 
				
				\draw [fill = gray!30] (9,0) circle (0.4); \node at (9,0) {$s_3$};
				\draw [fill = gray!30] (9,1.5) circle (0.4); \node at (9,1.5) {$s_2$};
				\draw [fill = gray!30] (9,3) circle (0.4); \node at (9,3) {$s_1$};
				\draw (11.5,1.5) circle (0.4); \node at (11.5,1.5) {$r_2$};
				\draw (11.5,3) circle (0.4); \node at (11.5,3) {$r_1$};
				\draw (9.4,3)--(11.1,3); 
				\node at (10.25,3.2) {$1/2$};
				\draw (9.4,1.5)--(11.1,3); 
				\node at (10,2.4) {$1/2$};
				\draw (9.4,1.5)--(11.1,1.5); 
				\node at (10.25,1.7) {$1/3$};
				\draw (9.4,0)--(11.1,1.5); 
				\node at (9.9,0.8) {$2/3$};
				\node at (10.25,-1) {\Ranking};
		\end{tikzpicture} }
	\end{center}
    \vspace{-10pt}
	\caption{Comparing OCS and \Ranking on a simple example graph, where the number next to an edge represents the probability that the edge is selected by the algorithm, conditioned on both offline neighbors being available when the corresponding online request arrives.}
	\label{fig:OCS-ranking-on-simple-graph}
\end{figure}

In other words, \Ranking favors servers at their first appearance (when they realize their ranks with ``fresh randomness''), while OCS acts the opposite.
Given that the vertices have bounded degrees, it is reasonable to believe that the strategy adopted by OCS is better.
Our result theoretically establishes this superiority of OCS over \Ranking, under the presence of degree bounds.

\subsection{Other Related Works}

Due to the vast literature in online bipartite matching, here we only review some of the most related works on degree-bounded graphs.
Buchbinder et al.~\cite{conf/esa/BuchbinderJN07} studied the Adwords problem when the maximum degree of the online vertices is $d$ and presented a $(1-(1-1/d)^d)$-competitive algorithm.
For the online stochastic matching, Bahmani and Kopralov~\cite{conf/esa/BahmaniK10} achieved a competitive ratio of $1-e^{-d}d^d/d!$ on $d$-regular graphs.
For online bipartite matching with random arrivals, Karande et al.~\cite{conf/stoc/KarandeMT11} showed that \Ranking achieves a competitive ratio of $1-O(\sqrt{1/d})$ when the graph admits $d$ disjoint perfect matchings, indicating the same ratio of \Ranking on $d$-regular graphs.
Cohen et al.~\cite{conf/focs/CohenPW19} designed a near-optimal algorithm for online edge-coloring, which translates into a ($1-o(1)$)-competitive algorithm for online bipartite matching on $(k,d)$-bounded graphs when $k$ is sufficiently large.
Albers and Schubert proposed optimal deterministic algorithms for online $b$-matching~\cite{conf/esa/AlbersS22} and Adwords~\cite{conf/wine/AlbersS22} on $(k,d)$-bounded graphs.
Cohen and Peng~\cite{conf/esa/CohenP23} proposed (optimal) fractional algorithms for online $b$-matching and Adwords problems on $(k,d)$-bounded graphs with $d\geq k$.
For further existing works on online matching problems, please refer to a recent survey by Huang et al.~\cite{huangonline}.

\section{Preliminaries}

We consider the following online bipartite matching problem.
Let $G(S\cup R, E)$ be the underlying bipartite graph, where the set of \emph{servers} $S$ is given offline while the \emph{requests} in $R$ arrive online one by one.
For each server $s\in S$ (resp. request $r\in R$), we use $N(s) \subseteq R$ (resp. $N(r) \subseteq S$) to denote its set of neighbors.
We use $d(s) = |N(s)|$ (resp. $d(r) = |N(r)|$) to denote the degree of server $s\in S$ (resp. request $r\in R$).
Upon the arrival of an online request $r\in R$, its set of neighboring servers $N(r)\subseteq S$ is revealed, and the online algorithm must either match $r$ to one of its unmatched neighbors $s\in N(r)$ irrevocably or leave $r$ unmatched permanently.
The objective is to maximize the size of the resulting matching.
We consider the problem on $(k,d)$-bounded graphs that are formally defined as follows.
        
\begin{definition}[$(k,d)$-bounded graphs]
    A bipartite graph $G(S\cup R, E)$ is called $(k,d)$-bounded if each offline server $s\in S$ has degree $d(s) \geq k$, and each online request $r\in R$ has degree $d(r) \leq d$.
\end{definition}

We assume that the online algorithm knows $k$ and $d$, but does not know the total number of online requests.
We use $\alg$ to denote the size of the matching produced by the online algorithm and use $\opt$ to denote the size of the maximum matching in hindsight.
Note that $\opt$ is uniquely determined by the input graph $G(S\cup R, E)$.
For randomized algorithms, $\alg$ is a random variable and we measure its performance by the competitive ratio, which is defined as the infimum of ${\E[\alg]}/{\opt} \in [0,1]$ over all online instances.
Note that for $(k,d)$-bounded graphs it is without loss of generality (w.l.o.g.) to assume that $d(s) = k$ for all servers $s\in S$, because we can artificially assume that only the first $k$ arrived requests in $N(s)$ are neighbors of $s$.
Note that this operation may decrease the degrees of online requests, but they are still upper bounded by $d$.

The main focus of this paper is on $(d,d)$-bounded graphs, which is a more general class of graphs than $d$-regular (in which all servers and requests have degree exactly $d$).
We show in the appendix that some of our results can be extended to $(k,d)$-bounded graphs with $k>d$.

\section{Upper Bound for the Competitive Ratio of Ranking} \label{sec:hardness-for-Ranking}

We first consider the \Ranking algorithm (see Algorithm~\ref{alg:Ranking}) proposed by Karp et al.~\cite{conf/stoc/KarpVV90}, which is formally described as follows.
We study the performance of \Ranking on $d$-regular graphs and show upper bounds for its competitive ratio for every $d\geq 2$.
Since $d$-regular graphs are $(d,d)$-bounded, the upper bounds also hold for $(d,d)$-bounded graphs.
We show in Appendix~\ref{appendix:kd} that similar upper bounds can be achieved for $(k,d)$-bounded graphs.

\begin{algorithm} \label{alg:Ranking}
\caption{Ranking}
\For{each server $s\in S$}{
    independently pick its rank $y_s\in [0,1]$ uniformly at random \;}
\For{each request $r\in R$}{
    match $r$ to an unmatched neighbor $s\in N(r)$ (if any) with the smallest rank $y_s$ \;}
\end{algorithm}

Cohen and Wajc~\cite{conf/soda/CohenW18} showed an upper bound of $0.875+o(1)$ for \Ranking on $d$-regular graphs when $d$ is sufficiently large.
For smaller $d$ (where $d$ is even), the upper bound is smaller (we will do a comparison at the end of this section).
In this section, we give better upper bounds on the competitive ratio of \Ranking for every $d\geq 2$.
Specifically, when $d=2$ the upper bound is $\frac{119}{144}\approx 0.8264$ while when $d\to \infty$ the upper bound is $1-\frac{1}{2e} \approx 0.8161$.


We first present a general hard instance for every $d\geq 2$, for which we can theoretically establish an upper bound of $1-\frac{d-1}{2d-1}\cdot \left(1-\frac{1}{d}\right)^d$ on the competitive ratio of \Ranking, which implies the $1-\frac{1}{2e}$ upper bound when $d\to \infty$.
However, we observe that the upper bound for the general instances with small $d$ is bad (e.g., it is $0.917$ and $0.882$ when $d= 2$ and $3$, respectively), and thus provide a modification of the general instance for $d=O(1)$.
We show that the competitive ratio of \Ranking is significantly lower, e.g., it is at most $0.827$ and $0.826$ when $d=2$ and $3$.

    
        
    
    

\subsection{Instance for General \texorpdfstring{$d$}{}} \label{sec:hardness-for-Ranking-general-d}

In this section, for every fixed $d\geq 2$, we present a hard instance of $d$-regular graphs and upper bound the competitive ratio of \Ranking on the instance.

\paragraph{Hard Instance for General $d$.}
Let the set of servers $S$ have size $2d-1$, and be partitioned into $S_1=\{s_1,\ldots,s_d\}$ and $S_2 = \{s_{d+1},\ldots,s_{2d-1}\}$.
Symmetrically, we let requests $R = R_1 \cup R_2$, where $R_1 = \{r_1,\ldots,r_d\}$ and $R_2 = \{r_{d+1},\ldots,r_{2d-1}\}$.
We add edges between $R_1$ and $S_2$, $R_2$ and $S_1$ to form two complete bipartite graphs.
Then for $i \in \{1,\dots,d\}$, let there be an edge between $s_i$ and $r_i$.
Let the requests arrive online in the order of $r_1,r_2,\ldots$. See Figure~\ref{fig:hardness-for-Ranking(d=5)} for an example with $d=5$.

\medskip

\begin{figure}[htb]
\begin{center}
\resizebox{0.5\textwidth}{!}{
\begin{tikzpicture}
\draw [fill = gray!30] (7.5, 0.75) circle (0.4); \node at (7.5,0.75) {$s_9$};
\draw [fill = gray!30] (7.5, 2.25) circle (0.4); \node at (7.5,2.25) {$s_8$};
\draw [fill = gray!30] (7.5, 3.75) circle (0.4); \node at (7.5,3.75) {$s_7$};
\draw [fill = gray!30] (7.5, 5.25) circle (0.4); \node at (7.5,5.25) {$s_6$};
\draw [fill = gray!30] (2.5,0) circle (0.4); \node at (2.5,0) {$s_5$};
\draw [fill = gray!30] (2.5,1.5) circle (0.4); \node at (2.5,1.5) {$s_4$};
\draw [fill = gray!30] (2.5,3) circle (0.4); \node at (2.5,3) {$s_3$};
\draw [fill = gray!30] (2.5,4.5) circle (0.4); \node at (2.5,4.5) {$s_2$};
\draw [fill = gray!30] (2.5,6) circle (0.4); \node at (2.5,6) {$s_1$};
\draw (0,0.75) circle (0.4); \node at (0,0.75) {$r_9$};
\draw (0,2.25) circle (0.4); \node at (0,2.25) {$r_8$};
\draw (0,3.75) circle (0.4); \node at (0,3.75) {$r_7$};
\draw (0,5.25) circle (0.4); \node at (0,5.25) {$r_6$};
\draw (5,0) circle (0.4); \node at (5,0) {$r_5$};
\draw (5,1.5) circle (0.4); \node at (5,1.5) {$r_4$};
\draw (5,3) circle (0.4); \node at (5,3) {$r_3$};
\draw (5,4.5) circle (0.4); \node at (5,4.5) {$r_2$};
\draw (5,6) circle (0.4); \node at (5,6) {$r_1$};
\draw (2.9,6)--(4.6,6); 
\draw (2.9,4.5)--(4.6,4.5); 
\draw (2.9,3)--(4.6,3); 
\draw (2.9,1.5)--(4.6,1.5); 
\draw (2.9, 0)--(4.6, 0); 
\draw (5.4, 6)--(7.1, 0.75);
\draw (5.4, 6)--(7.1, 2.25);
\draw (5.4, 6)--(7.1, 3.75);
\draw (5.4, 6)--(7.1, 5.25);
\draw (5.4, 4.5)--(7.1, 0.75);
\draw (5.4, 4.5)--(7.1, 2.25);
\draw (5.4, 4.5)--(7.1, 3.75);
\draw (5.4, 4.5)--(7.1, 5.25);
\draw (5.4, 3)--(7.1, 0.75);
\draw (5.4, 3)--(7.1, 2.25);
\draw (5.4, 3)--(7.1, 3.75);
\draw (5.4, 3)--(7.1, 5.25);
\draw (5.4, 1.5)--(7.1, 0.75);
\draw (5.4, 1.5)--(7.1, 2.25);
\draw (5.4, 1.5)--(7.1, 3.75);
\draw (5.4, 1.5)--(7.1, 5.25);
\draw (5.4, 0)--(7.1, 0.75);
\draw (5.4, 0)--(7.1, 2.25);
\draw (5.4, 0)--(7.1, 3.75);
\draw (5.4, 0)--(7.1, 5.25);
\draw (2.1, 6)--(0.4, 0.75);
\draw (2.1, 6)--(0.4, 2.25);
\draw (2.1, 6)--(0.4, 3.75);
\draw (2.1, 6)--(0.4, 5.25);
\draw (2.1, 4.5)--(0.4, 0.75);
\draw (2.1, 4.5)--(0.4, 2.25);
\draw (2.1, 4.5)--(0.4, 3.75);
\draw (2.1, 4.5)--(0.4, 5.25);
\draw (2.1, 3)--(0.4, 0.75);
\draw (2.1, 3)--(0.4, 2.25);
\draw (2.1, 3)--(0.4, 3.75);
\draw (2.1, 3)--(0.4, 5.25);
\draw (2.1, 1.5)--(0.4, 0.75);
\draw (2.1, 1.5)--(0.4, 2.25);
\draw (2.1, 1.5)--(0.4, 3.75);
\draw (2.1, 1.5)--(0.4, 5.25);
\draw (2.1, 0)--(0.4, 0.75);
\draw (2.1, 0)--(0.4, 2.25);
\draw (2.1, 0)--(0.4, 3.75);
\draw (2.1, 0)--(0.4, 5.25);
\draw [densely dashed] (-0.8,-0.05) rectangle (0.8,6.05); \node at (0,-1) {$R_2$};
\draw [densely dashed] (1.7,-0.8) rectangle (3.3,6.8); \node at (2.5,-1.5) {$S_1$};
\draw [densely dashed] (4.2,-0.8) rectangle (5.8,6.8); \node at (5,-1.5) {$R_1$};
\draw [densely dashed] (6.7,-0.05) rectangle (8.3,6.05); \node at (7.5,-1) {$S_2$};
\end{tikzpicture} }
\end{center}
\vspace{-20pt}
\caption{Hard instance for \Ranking: an illustrating example with $d=5$.}
\label{fig:hardness-for-Ranking(d=5)}
\end{figure}
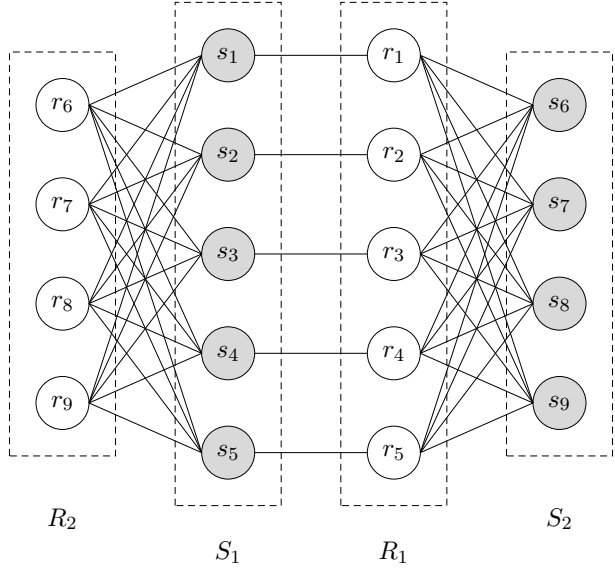

Intuitively speaking, following our observation from Figure~\ref{fig:OCS-ranking-on-simple-graph}, we show that the probability of online request $r_i$ (where $i\in \{1,\ldots,d\}$) choosing server $s_i$ becomes higher and higher, as the number of available servers in $S_2$ decreases.
We show that after the arrival of request $r_d$, a significant fraction of servers in $S_2$ are still unmatched (which cannot be matched anymore), and will be the reason for the underperformance of \Ranking.
Specifically, we show that after the arrival of request $r_i$, about $(1-\frac{1}{d})^i$ fraction of servers in $S_2$ remain unmatched in expectation, which leads to the following lower bound on the competitive ratio of \Ranking.

\begin{theorem} \label{theorem:upper-bound-for-Ranking}
    \textnormal{\Ranking} is at most $\gamma(d)$-competitive for online bipartite matching on $d$-regular graphs, where $\gamma(d) = 1-\frac{d-1}{2d-1}\cdot \left(1-\frac{1}{d}\right)^d$.
\end{theorem}
    
\begin{corollary}
    \textnormal{\Ranking} is at most $0.8161$-competitive for online bipartite matching on $d$-regular graphs when $d$ is sufficiently large. 
\end{corollary}

Notice that servers in $S_1$ always get matched by \Ranking while those in $S_2$ are not.
The nature of \Ranking makes the competing environment of the servers in $S_2$ get worse and worse since those with small ranks will be matched early and those remaining servers have to compete with a fresh random rank drawn from $U[0,1]$ when a request in $R_1$ arrives.
To characterize the performance of \Ranking, we provide the following equivalent description of the algorithm as follows.
For notational convenience, in the following, we use $y_i$ to denote the rank of server $s_i$.
Moreover, we assume w.l.o.g. that all ranks are different.

\paragraph{Equivalent Description for Ranking.}
We first fix the ranks of servers in $S_2$ and suppose that the ordered ranks are $0\leq \theta_1 < \theta_2 < \cdots < \theta_{d-1}\leq 1$ (sorted version of $y_{d+1},y_{d+2},\ldots,y_{2d-1}$).
Upon the arrival of request $r_i$, we realize the rank $y_i$ of $s_i$.
For example, if $j$ servers in $S_2$ are matched when $r_i$ arrives, then $r_i$ matches $s_i$ if and only if $y_i\in [0,\theta_{j+1})$.
The performance of \Ranking will then be measured after first taking an expectation over $y=(y_1,\ldots,y_d)$, and then taking an expectation over $\theta=(\theta_1,\ldots,\theta_{d-1})$.

\medskip

We introduce the following definition to characterize the number of matched servers in $S_2$.

\begin{definition}
    Let $X(i,j)$ be the random variable denoting the number of servers in $S_2$ that are matched by $\{r_i,r_{i+1},\ldots,r_d\}$, conditioned on $\theta_j$ being the minimum rank of unmatched servers in $S_2$ when $r_i$ arrives.
    For completeness, we define $X(d+1,\cdot) = 0$.
\end{definition}

We can equivalently interpret $X(1,1)$ as a result of the following Markov chain (see Figure~\ref{fig:Markov-chain}) with $d$ states (where state $d-j$ corresponds to the situation when $d-j$ servers in $S_2$ are unmatched).
Initially, we are at the original state $d-1$.
Conditioned on being at state $d-j$, the arrival of a request $r_i$ corresponds to taking one step in the Markov chain.
As a consequence, we stay at state $d-j$ with probability $\theta_j$ (corresponding to the case when $y_i\in [0,\theta_j)$); move to state $d-j-1$ with probability $1-\theta_j$ (corresponding to the case when $y_i\in (\theta_j,1]$).
Therefore, $X(1,1)$ can be viewed as the distance we have moved from the origin after $d$ steps; $X(i,j)$ can be viewed as the distance we have moved, starting from state $d-j$, after taking $d-i+1$ steps.

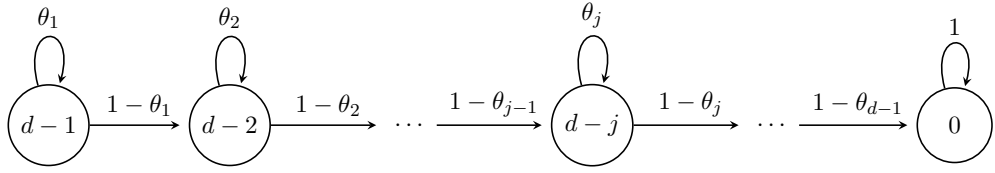
\begin{figure}[htb]
\centering
\resizebox{0.8\textwidth}{!}{
\begin{tikzpicture}[->, >= stealth, shorten >= 4pt, line width = 0.7pt, node distance = 2.8cm]
    \node [circle, draw, minimum size = 3em] (1) {$d-1$};
    \node [circle, draw, minimum size = 3em] (2) [right of = 1] {$d-2$};
    \node (3) [right of = 2] {$\cdots$};
    \node [circle, draw, minimum size = 3em] (4) [right of = 3] {$d-j$};
    \node (5) [right of = 4] {$\cdots$};
    \node [circle, draw, minimum size = 3em] (6) [right of = 5] {$0$};
    \path (1) edge [loop above] node {$\theta_1$} (1);
    \path (1) edge node[above] {$1-\theta_1$} (2);
    \path (2) edge [loop above] node {$\theta_2$} (2);
    \path (2) edge node[above] {$1-\theta_2$} (3);
    \path (3) edge node[above] {$1-\theta_{j-1}$} (4);
    \path (4) edge [loop above] node {$\theta_j$} (4);
    \path (4) edge node[above] {$1-\theta_{j}$} (5);
    \path (5) edge node[above] {$1-\theta_{d-1}$} (6);
    \path (6) edge [loop above] node {$1$} (6);
\end{tikzpicture}}
\caption{An Markov chain based illustration of $X(i,j)$.
}
\label{fig:Markov-chain}
\end{figure}


Therefore, we can express $\E_{y}[X(i,j)]$ as a function of $\theta$.
Particularly, we use $\mathcal{F}(\theta_1,\ldots,\theta_{d-1})$ to denote $\E_{y}[X(1,1)]$, the expected number of servers in $S_2$ \Ranking matches.
Recall that all servers in $S_1$ always get matched by \Ranking.
The competitive ratio of \Ranking in our instance is then measured by
\begin{equation*}
    \E_{\theta}\left[ \frac{d+\E_{y}[X(1,1)]}{2d-1} \right] = \frac{d+\E_{\theta}[\mathcal{F}(\theta_1,\ldots,\theta_{d-1})]}{2d-1}.
\end{equation*}

The following lemma states that the competitive ratio of \Ranking for the above instance will not decrease if we fix $\theta$ to be $(\frac{1}{d},\ldots,\frac{d-1}{d})$, which is crucial in proving Theorem~\ref{theorem:upper-bound-for-Ranking}.

\begin{lemma} \label{lemma:concavity-of-F}
    We have
    \begin{equation*}
        \E_{\theta}\left[ \mathcal{F}(\theta_1,\ldots,\theta_{d-1}) \right] \leq \mathcal{F}\left(\frac{1}{d},\ldots,\frac{d-1}{d}\right).
    \end{equation*}
\end{lemma}

Before proving Lemma~\ref{lemma:concavity-of-F}, we first prove a weaker statement showing the concavity of $\E_{y}[X(1,1)]$ (a.k.a. $\mathcal{F}(\theta_1,\ldots,\theta_{d-1})$) in every dimension $\theta_j$.
Note that $\E_{y}[X(1,1)]$ is decreasing when $\theta_j$ increases (as the probability of staying at state $d-j$ increases).

\begin{lemma} \label{lemma:concave-in-theta}
    For any $j\in\{1,2,\ldots,d-1\}$, $\E_{y}[X(1,1)]$ is concave in $\theta_j$.
\end{lemma}
\begin{proof}
    With a slight abuse of notation, we also use $\theta_j$ to denote the server in $S_2$ with rank $\theta_j$.
    Fix any $j\in\{1,2,\ldots,d-1\}$, depending on the number of steps it takes to reach state $d-j$ (starting from state $d-1$), we can express $\E_y[X(1,1)]$ as
    \begin{align*}
        \E_y[X(1,1)] & = \sum_{i=j}^d \left( \Pr[r_{i-1} \text{ matches } \theta_{j-1}]\cdot \left(\E_y[X(i,j)] + (j-1) \right) \right) \\
        & \qquad + \Pr[\theta_{j-1} \text{ is unmatched}]\cdot \E_y[X(1,1) \mid \theta_{j-1} \text{ is unmatched}].
    \end{align*}

    Since in the above equation, only the terms $\E_y[X(\cdot,j)]$ depend on $\theta_j$, to show that $\E_{y}[X(1,1)]$ is concave in $\theta_j$, it suffices to argue that $\E_{y}[X(i,j)]$ is concave in $\theta_j$ for every $i\in \{ j,j+1,\ldots,d\}$.

    Fix any $i\in \{ j,j+1,\ldots,d\}$, we have
    \begin{align*}
        \E_y[X(i,j)] & = \sum_{l=i}^{d}\left( \Pr[r_{l} \text{ matches } \theta_j]\cdot (\E_y[X(l+1,j+1)] + 1) \right) \\
        & = \sum_{l=i}^{d}\left( \theta_j^{l-i}(1-\theta_j) \cdot \E_y[X(l+1,j+1)] \right) + (1-\theta_j^{d+1-i}).
    \end{align*}

    Notice that $\E_y[X(l+1,j+1)]$ is decreasing in $l$. We have
    \begin{align*}
        \E_y[X(i,j)] & = \sum_{l=i+1}^{d}\left( \left( \E_y[X(l,j+1)]-\E_y[X(l+1,j+1)] \right) \cdot \sum_{t=i}^{l-1}(\theta_j^{t-i}(1-\theta_j)) \right) + (1-\theta_j^{d+1-i}) \\
        & = \sum_{l=i+1}^{d}\left( \left( \E_y[X(l,j+1)]-\E_y[X(l+1,j+1)] \right) \cdot (1-\theta_j^{l-i}) \right) + (1-\theta_j^{d+1-i}).
    \end{align*}

    Note that $\E_y[X(l,j+1)]-\E_y[X(l+1,j+1)]$ is independent of $\theta_j$. Therefore, it can be verified that the partial derivative of $\E_y[X(i,j)]$ on $\theta_j$ is non-positive and decreasing, which implies that $\E_y[X(i,j)]$ is concave in $\theta_j$. 
\end{proof}

In the following, we construct a random vector $g=(g_1,g_2,\ldots,g_{d-1})$ and show that it is statistically identical to $\theta = (\theta_1,\theta_2,\ldots,\theta_{d-1})$ (a proof is provided in Appendix~\ref{sec:missing-proofs(Ranking)}).

\begin{definition} \label{definition:m-and-g}
    Let $0\leq g_1 \leq g_2 \leq \cdots \leq g_{d-1} \leq 1$ be $d-1$ ordered random variables constructed as follows.
    For each $i\in \{1,\ldots,d-1\}$:
    \begin{itemize}
        \item let $m_{i}$ be the maximum of $i$ independent random variables $z_1^{(i)},\ldots,z_i^{(i)}$ that are drawn from $U[0,1]$, i.e., $m_i = \max\{z_1^{(i)},\ldots,z_i^{(i)}\}$;
        \item define $g_i = \prod_{j=i}^{d-1} m_j$.
    \end{itemize}
\end{definition}

Note that $m_1,\ldots,m_{d-1}$ are independent random variables.
Besides, we have $\E[m_i]=i/(i+1)$ for any $i\in \{1,\ldots,d-1\}$, since by definition we have $\Pr[m_i<t] = t^i$ for any $t\in [0,1]$.
Below, we let $m=(m_1,\ldots,m_{d-1})$.

\begin{proofof}{Lemma~\ref{lemma:concavity-of-F}}
    By the stochastic equivalence between $g$ and $\theta$, we have (note that $\mathcal{F}$ can also be regarded as a function of $m$ as it defines $g$.)
    \begin{equation*}
        \E_\theta[\mathcal{F}(\theta_1,\ldots,\theta_{d-1})] = \E_g[\mathcal{F}(g_1,\ldots,g_{d-1})] = \E_m[\mathcal{F}(g_1,\ldots,g_{d-1})].
    \end{equation*}

    Next, we show that $\mathcal{F}$ is concave in $m_j$ for every $j\in \{1,\ldots,d-1\}$.
    Taking partial derivative of $\mathcal{F}$ on $m_j$, we have
    \begin{equation*}
         \frac{\partial \mathcal{F}}{\partial m_j} = \sum_{i=1}^{d-1} \frac{\partial \mathcal{F}}{\partial g_i} \cdot \frac{\partial g_i}{\partial m_j}. 
    \end{equation*}
    
    Notice that in the above equation, the term $\frac{\partial g_i}{\partial m_j}$ is non-negative\footnote{Specifically, when $i > j$ we have $\frac{\partial g_i}{\partial m_j} = 0$; when $i \leq j$ we have $\frac{\partial g_i}{\partial m_j} = \frac{g_i}{m_j}\geq 0$.} and independent of $m_j$, for any $i\in \{1,\ldots,d-1\}$.
    Moreover, $\frac{\partial \mathcal{F}}{\partial g_i}$ is decreasing in $g_i$ for each $i\in \{1,\ldots,d-1\}$ by Lemma~\ref{lemma:concave-in-theta}.
    Since $g_i$ is linear in $m_j$ with a non-negative coefficient, $\frac{\partial \mathcal{F}}{\partial g_i}$ is non-increasing in $m_j$.
    Therefore, $\frac{\partial \mathcal{F}}{\partial m_j}$ is non-increasing, which further indicates the concavity of $\mathcal{F}$ in $m_j$ for every $j\in \{1,\ldots,d-1\}$.

    By the concavity of $\mathcal{F}$ in $m_{d-1}$ and Jensen's Inequality, we have
    \begin{align*}
        & \ \E_m[\mathcal{F}(g_1,\ldots,g_{d-2},g_{d-1})] = \E_m\left[\mathcal{F}\left(\prod_{i=1}^{d-1}m_i,\ldots,\prod_{i=d-2}^{d-1}m_i,m_{d-1}\right)\right] \\
        = & \ \E_{m_1,\ldots,m_{d-2}}\left[\E_{m_{d-1}}\left[\mathcal{F}\left(\prod_{i=1}^{d-1}m_i,\ldots,\prod_{i=d-2}^{d-1}m_i,m_{d-1}\right)\right]\right] \\
        \leq & \ \E_{m_1,\ldots,m_{d-2}}\left[\mathcal{F}\left(\E[m_{d-1}] \cdot \prod_{i=1}^{d-2}m_i,\ldots, \E[m_{d-1}] \cdot m_{d-2},\E[m_{d-1}]\right)\right] \\
        = & \ \E_{m_1,\ldots,m_{d-2}}\left[\mathcal{F}\left(\frac{d-1}{d} \cdot \prod_{i=1}^{d-2}m_i,\ldots,\frac{d-1}{d} \cdot m_{d-2},\frac{d-1}{d}\right)\right].
    \end{align*}
    
    Similarly, by the concavity of $\mathcal{F}$ on $m_j$ for $j\in \{1,\ldots,d-2\}$, we have
    \begin{align*}
        & \ \E_{m_1,\ldots,m_{d-2}}\left[\mathcal{F}\left(\frac{d-1}{d}\prod_{i=1}^{d-2}m_i,\ldots,\frac{d-1}{d}\cdot m_{d-2},\frac{d-1}{d}\right)\right] \\ 
        \leq & \ \E_{m_1,\ldots,m_{d-3}}\left[\mathcal{F}\left(\frac{d-1}{d}\cdot \frac{d-2}{d-1} \prod_{i=1}^{d-3}m_i,\ldots,\frac{d-1}{d}\cdot \frac{d-2}{d-1},\frac{d-1}{d}\right)\right] \\ 
        = & \ \E_{m_1,\ldots,m_{d-3}}\left[\mathcal{F}\left(\frac{d-2}{d}\prod_{i=1}^{d-3}m_i,\ldots,\frac{d-2}{d},\frac{d-1}{d}\right)\right] \leq \mathcal{F}\left(\frac{1}{d},\cdots,\frac{d-2}{d},\frac{d-1}{d}\right).
    \end{align*}

    Therefore, we have $\E_\theta[\mathcal{F}(\theta_1,\ldots,\theta_{d-1})] \leq \mathcal{F}\left(\frac{1}{d},\ldots,\frac{d-1}{d}\right)$, which finishes the proof.
\end{proofof}

Now we have all the ingredients to prove Theorem~\ref{theorem:upper-bound-for-Ranking}.

\begin{proofof}{Theorem~\ref{theorem:upper-bound-for-Ranking}}
    Recall that we have upper bounded the expected number of matched servers by $d+\mathcal{F}\left(\frac{1}{d},\ldots,\frac{d-1}{d}\right)$.
    In the following we further give an upper bound on $\mathcal{F}\left(\frac{1}{d},\ldots,\frac{d-1}{d}\right)$ in terms of $d$.
    Let $\beta(0)=d-1$.
    For each $i\in \{1,\dots,d\}$, let $\beta(i)$ be a random variable denoting the number of servers in $S_2$ that are unmatched at the end of the round when $r_i$ arrives, conditioned on $\theta=(\frac{1}{d},\ldots,\frac{d-1}{d})$.
    Recall that $r_i$ matches $s_i$ if and only if $y_i$ (the rank of $s_i$) is smaller than $\frac{d - \beta(i-1)}{d}$.
    Therefore we have
    \begin{equation*}
        \beta(i) = 
        \begin{cases}
            \beta(i-1), \ & \text{w.p.} \ 1-\frac{\beta(i-1)}{d}; \\
            \beta(i-1)-1, \ & \text{w.p.} \ \frac{\beta(i-1)}{d}.
        \end{cases}
    \end{equation*}

    Therefore, conditioned on the value of $\beta(i-1)$, we have
    \begin{equation*}
        \E_{y_i}[\beta(i) \mid \beta(i-1)] = \beta(i-1) - \frac{\beta(i-1)}{d} = \frac{d-1}{d}\cdot \beta(i-1).
    \end{equation*}

    Taking expectation of $\beta(i-1)$ (over $y_1,\ldots,y_{i-1}$), we have
    \begin{align*}
        \E_{y_1,\ldots,y_i}[\beta(i)] = \frac{d-1}{d} \cdot \E_{y_1,\ldots,y_{i-1}}[\beta(i-1)].
    \end{align*}
    
    Therefore, the expected number of servers in $S_2$ that are unmatched in the end is
    \begin{align*}
        \E_y[\beta(d)] & = \frac{d-1}{d} \cdot \E_{y_1,\ldots,y_{d-1}}[\beta(d-1)] \\
        & = \left(\frac{d-1}{d}\right)^2 \cdot \E_{y_1,\ldots,y_{d-2}}[\beta(d-2)] = \left(1-\frac{1}{d}\right)^d \cdot (d-1).
    \end{align*}
    
    Therefore the expected number of matched servers in $S_2$ is $\mathcal{F}(\frac{1}{d},\ldots,\frac{d-1}{d}) = d-1-\E_y[\beta(d)]$, and the competitive ratio of \Ranking for the hard instance is at most
    \begin{align*}
        \E_{\theta}\left[ \frac{d+\mathcal{F}(\theta_1,\ldots,\theta_{d-1})}{2d-1} \right] &\leq \frac{d+\mathcal{F}(\frac{1}{d},\ldots,\frac{d-1}{d})}{2d-1} = \frac{2d-1-\E_y[\beta(d)]}{2d-1} \\
        &= 1- \frac{d-1}{2d-1}\cdot \left(1-\frac{1}{d}\right)^d,
    \end{align*}
    where the inequality holds due to Lemma~\ref{lemma:concavity-of-F}.
\end{proofof}

\subsection{Stronger Upper Bound for Small \texorpdfstring{$d$}{}}
\label{ssec:improved-ranking}


While we can show that $\lim_{d\to \infty} \gamma(d) = 1-\frac{1}{2e}\approx 0.816$ and $\gamma(d) \leq 0.826$ when $d\geq 20$, the ratio is large when $d$ is small.
In this section, we improve our previous instance and obtain smaller upper bounds for \Ranking when $d=O(1)$.
Consider the following hard instance (see Figure~\ref{fig:hardness-for-Ranking(d=3)} for an example when $d=3$ and Figure~\ref{fig:hardness-for-Ranking(d=2)} (left) for an example when $d=2$).

\paragraph{Hard Instance for Small $d$.}
The instance is consist of $2d$ components $G_1,\ldots, G_{2d}$.
For each $i\in \{1,2,\ldots,d\}$, let component $G_i$ have $4$ levels, the first being a request $r_0^i$, the second having $d$ servers $\{s_1^i,\ldots,s_d^i\}$, the third having $d$ requests $\{r_1^i,\ldots,r_d^i\}$, and the fourth having $d$ servers $\{s_{d+1}^i,\ldots,s_{2d}^i\}$.
Each component $G_{d+i}$ (where $i\in \{1,\ldots,d\}$) is consist of $d-1$ requests $\{r_1^{d+i},\ldots,r_{d-1}^{d+i}\}$.
For each $i\in \{1,\ldots,d\}$, let $r_0^i$ be a common neighbor of $\{s_1^i,\ldots,s_d^i\}$.
For all $i,j\in \{1,\ldots,d\}$,
\begin{itemize}
    \item let the neighbors of $s_j^i$ be $\{r_1^i,\ldots,r_d^i\}\setminus\{r_{j'}^i\}$, where $j' = (j \bmod d) + 1$;
    \item let there be an edge between $s_{d+j}^i$ and $r_j^i$.
\end{itemize}
Finally, for each $l\in \{d+1,\ldots,2d\}$, we form a complete bipartite graph between $G_l$ and $\{s_l^1,\ldots,s_l^d\}$. 
We let requests arrive in the order that $r_i^j$ arrives before $r_{i'}^{j'}$ if either $j<j'$, or $j=j'$ and $i<i'$.

\medskip

\begin{figure}[htb]
\begin{center}
\resizebox{0.8\textwidth}{!}{
\begin{tikzpicture}
\draw [fill = gray!30] (0,6) circle (0.4); \node at (0,6) {$s_1^1$};
\draw [fill = gray!30] (2,6) circle (0.4); \node at (2,6) {$s_2^1$};
\draw [fill = gray!30] (4,6) circle (0.4); \node at (4,6) {$s_3^1$};
\draw [fill = gray!30] (0,2) circle (0.4); \node at (0,2) {$s_4^1$};
\draw [fill = gray!30] (2,2) circle (0.4); \node at (2,2) {$s_5^1$};
\draw [fill = gray!30] (4,2) circle (0.4); \node at (4,2) {$s_6^1$};
\draw [fill = gray!30] (6,6) circle (0.4); \node at (6,6) {$s_1^2$};
\draw [fill = gray!30] (8,6) circle (0.4); \node at (8,6) {$s_2^2$};
\draw [fill = gray!30] (10,6) circle (0.4); \node at (10,6) {$s_3^2$};
\draw [fill = gray!30] (6,2) circle (0.4); \node at (6,2) {$s_4^2$};
\draw [fill = gray!30] (8,2) circle (0.4); \node at (8,2) {$s_5^2$};
\draw [fill = gray!30] (10,2) circle (0.4); \node at (10,2) {$s_6^2$};
\draw [fill = gray!30] (12,6) circle (0.4); \node at (12,6) {$s_1^3$};
\draw [fill = gray!30] (14,6) circle (0.4); \node at (14,6) {$s_2^3$};
\draw [fill = gray!30] (16,6) circle (0.4); \node at (16,6) {$s_3^3$};
\draw [fill = gray!30] (12,2) circle (0.4); \node at (12,2) {$s_4^3$};
\draw [fill = gray!30] (14,2) circle (0.4); \node at (14,2) {$s_5^3$};
\draw [fill = gray!30] (16,2) circle (0.4); \node at (16,2) {$s_6^3$};
\draw (2,8) circle (0.4); \node at (2,8) {$r_0^1$};
\draw (0,4) circle (0.4); \node at (0,4) {$r_1^1$};
\draw (2,4) circle (0.4); \node at (2,4) {$r_2^1$};
\draw (4,4) circle (0.4); \node at (4,4) {$r_3^1$};
\draw (8,8) circle (0.4); \node at (8,8) {$r_0^2$};
\draw (6,4) circle (0.4); \node at (6,4) {$r_1^2$};
\draw (8,4) circle (0.4); \node at (8,4) {$r_2^2$};
\draw (10,4) circle (0.4); \node at (10,4) {$r_3^2$};
\draw (14,8) circle (0.4); \node at (14,8) {$r_0^3$};
\draw (12,4) circle (0.4); \node at (12,4) {$r_1^3$};
\draw (14,4) circle (0.4); \node at (14,4) {$r_2^3$};
\draw (16,4) circle (0.4); \node at (16,4) {$r_3^3$};
\draw (1,0) circle (0.4); \node at (1,0) {$r_1^4$};
\draw (3,0) circle (0.4); \node at (3,0) {$r_2^4$};
\draw (7,0) circle (0.4); \node at (7,0) {$r_1^5$};
\draw (9,0) circle (0.4); \node at (9,0) {$r_2^5$};
\draw (13,0) circle (0.4); \node at (13,0) {$r_1^6$};
\draw (15,0) circle (0.4); \node at (15,0) {$r_2^6$};
\draw (0,6.4)--(2,7.6); 
\draw (2,6.4)--(2,7.6); 
\draw (4,6.4)--(2,7.6); 
\draw (0,5.6)--(0,4.4); 
\draw (0,5.6)--(4,4.4); 
\draw (2,5.6)--(2,4.4); 
\draw (2,5.6)--(0,4.4); 
\draw (4,5.6)--(4,4.4); 
\draw (4,5.6)--(2,4.4); 
\draw (0,2.4)--(0,3.6); 
\draw (2,2.4)--(2,3.6); 
\draw (4,2.4)--(4,3.6); 
\draw (6,6.4)--(8,7.6); 
\draw (8,6.4)--(8,7.6); 
\draw (10,6.4)--(8,7.6); 
\draw (6,5.6)--(6,4.4); 
\draw (6,5.6)--(10,4.4); 
\draw (8,5.6)--(8,4.4); 
\draw (8,5.6)--(6,4.4); 
\draw (10,5.6)--(10,4.4); 
\draw (10,5.6)--(8,4.4); 
\draw (6,2.4)--(6,3.6); 
\draw (8,2.4)--(8,3.6); 
\draw (10,2.4)--(10,3.6); 
\draw (12,6.4)--(14,7.6); 
\draw (14,6.4)--(14,7.6); 
\draw (16,6.4)--(14,7.6); 
\draw (12,5.6)--(12,4.4); 
\draw (12,5.6)--(16,4.4); 
\draw (14,5.6)--(14,4.4); 
\draw (14,5.6)--(12,4.4); 
\draw (16,5.6)--(16,4.4); 
\draw (16,5.6)--(14,4.4); 
\draw (12,2.4)--(12,3.6); 
\draw (14,2.4)--(14,3.6); 
\draw (16,2.4)--(16,3.6); 
\draw (0,1.6)--(1,0.4); 
\draw (0,1.6)--(3,0.4); 
\draw (6,1.6)--(1,0.4); 
\draw (6,1.6)--(3,0.4); 
\draw (12,1.6)--(1,0.4); 
\draw (12,1.6)--(3,0.4); 
\draw (2,1.6)--(7,0.4); 
\draw (2,1.6)--(9,0.4); 
\draw (8,1.6)--(7,0.4); 
\draw (8,1.6)--(9,0.4); 
\draw (14,1.6)--(7,0.4); 
\draw (14,1.6)--(9,0.4); 
\draw (4,1.6)--(13,0.4); 
\draw (4,1.6)--(15,0.4); 
\draw (10,1.6)--(13,0.4); 
\draw (10,1.6)--(15,0.4); 
\draw (16,1.6)--(13,0.4); 
\draw (16,1.6)--(15,0.4); 
\node at (-2,8) {\large{\textbf{level 1}}};
\node at (-2,6) {\large{\textbf{level 2}}};
\node at (-2,4) {\large{\textbf{level 3}}};
\node at (-2,2) {\large{\textbf{level 4}}};
\node at (-2,0) {\large{\textbf{level 5}}};
\draw [densely dashed] (-0.7,8.7) rectangle (4.7,1.3); \node at (4,8) {$G_1$};
\draw [densely dashed] (0,-0.7) rectangle (4,0.7); \node at (-0.5,0) {$G_4$};
\end{tikzpicture}}
\end{center}
\caption{Hard instance for \Ranking: an illustrating example when $d=3$.}
\label{fig:hardness-for-Ranking(d=3)}
\end{figure}
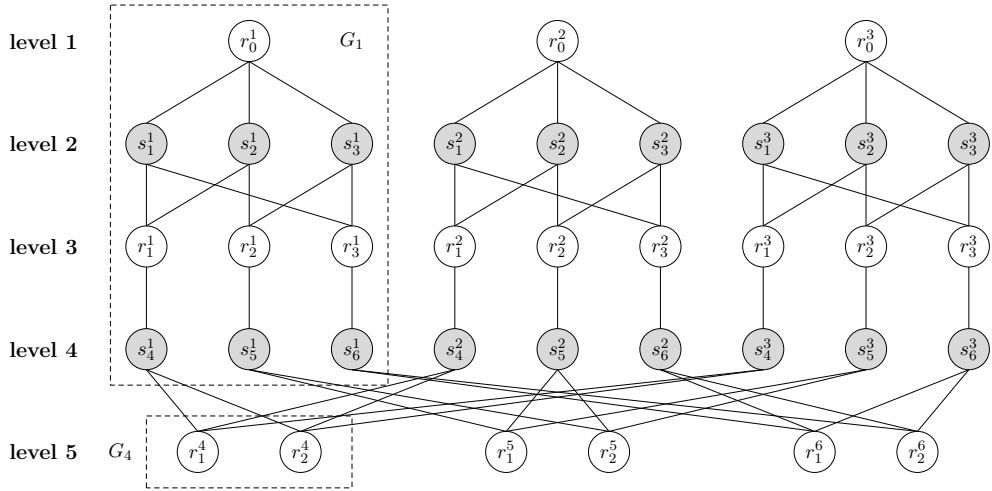

Intuitively speaking, the design of the hard instance is to increase the probability of servers in level $4$ being chosen before the requests in $G_{d+1},\ldots, G_{2d}$ arrive.
Take $s^1_6$ in Figure~\ref{fig:hardness-for-Ranking(d=3)} as an example.
Upon the arrival of $r^1_3$, $s^1_6$ realizes its rank and will be matched if its rank is smaller than that of $s^1_1$ and $s^1_3$ if they are not unmatched.
Notice that this ($s^1_6$ being matched by $r^1_3$) is very likely to happen because either $s^1_1$ and $s^1_3$ are already matched (by $r^1_0$, $r^1_1$ or $r^1_2$), or they have large ranks.
Similarly, we can argue that $s^2_6$ and $s^3_6$ will very likely be matched before $r^6_1$ and $r^6_2$ arrive, which results in $r^6_1$ and $r^6_2$ being unmatched in \Ranking.

For small $d$ we can precisely calculate the probability of each server in level $4$ being matched before the requests in $G_{d+1},\ldots, G_{2d}$ arrive, and use it to bound the competitive ratio of \Ranking.
We showcase the analysis when $d=2$ (see Figure~\ref{fig:hardness-for-Ranking(d=2)} (left) for the instance) in the following (the proof can be found in Appendix~\ref{sec:missing-proofs(Ranking)}).

\begin{lemma} \label{lemma:hardness-for-Ranking(d=2)}
    \textnormal{\Ranking} is at most $\frac{119}{144} (\approx 0.8264)$-competitive for online bipartite matching on $2$-regular graphs.
\end{lemma}

We remark that it would be quite tedious (and difficult, in some cases) to analyze the competitive ratio of \Ranking in the above instance for $d \geq 3$.
Therefore, we calculate the ratio for larger $d$ with the assistance of a computer program.
Observe that prior to the arrival of level $5$ requests, the matching status of vertices in components $G_1,\ldots,G_d$ are independent.
Therefore to compute the competitive ratio, it suffices to calculate the probability of each server in level $4$ being matching, which can be done by enumerating all $(2d)!$ permutations of servers' ranks in $G_1$.
Due to the complexity of the enumeration, we compute the competitive ratio of \Ranking for the above instances and use it as an upper bound for \Ranking only for $2\leq d\leq 6$ (see Table~\ref{fig:improved-hardness-for-Ranking-(d,d)}).
For $d\geq 7$, we use $\gamma(d)$ in Theorem~\ref{theorem:upper-bound-for-Ranking} as an upper bound.

\begin{table}[htb]
    \begin{center}
        \begin{tabular}{|c|c|c|c|c|c|c} 
           \hline
           $d$ & $2$ & $3$ & $4$ & $5$ & $6$ \\
           \hline
           U.B. & $0.8264$ & $0.8251$ & $0.8228$ & $0.8223$ & $0.8219$ \\
           \hline
        \end{tabular}
    \end{center}
    \caption{Upper bounds on the competitive ratio of \Ranking when $d\in \{2,3,4,5,6\}$.}
    \label{fig:improved-hardness-for-Ranking-(d,d)}
\end{table}

    
        

\paragraph{Comparison with \cite{conf/soda/CohenW18}.}
We briefly compare our hard instance with that of \cite{conf/soda/CohenW18}.
Consider the case when $d=2$ (see Figure~\ref{fig:hardness-for-Ranking(d=2)} (right) for their hard instance, where requests of both instances arrive in the order of $r_1,r_2,\ldots$.).
For the two instances, the competitive ratios are $0.8264$ and $0.8333$.
In the right instance, the event of $s_3$ being matched is negatively correlated with that of $s_4$ being matched when $r_4$ arrives.
In the left instance, the event of $s_3$ being matched is independent of that of $s_7$ being matched (similarly for $s_4$ and $s_8$), which increases the probability of $r_7$ and $r_8$ being unmatched and is the reason why the eventual ratio is smaller.
The way they generalized the instance to larger $d$'s is also different from ours: they replaced each vertex with $d/2$ identical copies and each edge with a complete bipartite graph between the two sets of copies (which only works for even $d$).
It can thus be shown that (when $d\to \infty$) the competitive ratio of \Ranking is upper bounded by that of the \textsc{Balance} algorithm~\cite{journals/tcs/KalyanasundaramP00} on the right instance in Figure~\ref{fig:hardness-for-Ranking(d=2)}, which is $0.875$.

\begin{figure}[htb]
\begin{center}
\resizebox{0.7\textwidth}{!}{
\begin{tikzpicture}
\draw [fill = gray!30] (0,6) circle (0.4); \node at (0,6) {$s_1$};
\draw [fill = gray!30] (3,6) circle (0.4); \node at (3,6) {$s_2$};
\draw [fill = gray!30] (0,2) circle (0.4); \node at (0,2) {$s_3$};
\draw [fill = gray!30] (3,2) circle (0.4); \node at (3,2) {$s_4$};
\draw [fill = gray!30] (6,6) circle (0.4); \node at (6,6) {$s_5$};
\draw [fill = gray!30] (9,6) circle (0.4); \node at (9,6) {$s_6$};
\draw [fill = gray!30] (6,2) circle (0.4); \node at (6,2) {$s_7$};
\draw [fill = gray!30] (9,2) circle (0.4); \node at (9,2) {$s_8$};
\draw (1.5,8) circle (0.4); \node at (1.5,8) {$r_1$};
\draw (0,4) circle (0.4); \node at (0,4) {$r_2$};
\draw (3,4) circle (0.4); \node at (3,4) {$r_3$};
\draw (7.5,8) circle (0.4); \node at (7.5,8) {$r_4$};
\draw (6,4) circle (0.4); \node at (6,4) {$r_5$};
\draw (9,4) circle (0.4); \node at (9,4) {$r_6$};
\draw (1.5,0) circle (0.4); \node at (1.5,0) {$r_7$};
\draw (7.5,0) circle (0.4); \node at (7.5,0) {$r_8$};
\draw (0,6.4)--(1.5,7.6); 
\draw (3,6.4)--(1.5,7.6); 
\draw (0,5.6)--(0,4.4); 
\draw (3,5.6)--(3,4.4); 
\draw (0,2.4)--(0,3.6); 
\draw (3,2.4)--(3,3.6); 
\draw (6,6.4)--(7.5,7.6); 
\draw (9,6.4)--(7.5,7.6); 
\draw (6,5.6)--(6,4.4); 
\draw (9,5.6)--(9,4.4); 
\draw (6,2.4)--(6,3.6); 
\draw (9,2.4)--(9,3.6); 
\draw (0,1.6)--(1.5,0.4); 
\draw (6,1.6)--(1.5,0.4); 
\draw (3,1.6)--(7.5,0.4); 
\draw (9,1.6)--(7.5,0.4); 
\draw [fill = gray!30] (13,6) circle (0.4); \node at (13,6) {$s_1$};
\draw [fill = gray!30] (16,6) circle (0.4); \node at (16,6) {$s_2$};
\draw [fill = gray!30] (13,2) circle (0.4); \node at (13,2) {$s_3$};
\draw [fill = gray!30] (16,2) circle (0.4); \node at (16,2) {$s_4$};
\draw (14.5,8) circle (0.4); \node at (14.5,8) {$r_1$};
\draw (13,4) circle (0.4); \node at (13,4) {$r_2$};
\draw (16,4) circle (0.4); \node at (16,4) {$r_3$};
\draw (14.5,0) circle (0.4); \node at (14.5,0) {$r_4$};
\draw [densely dashed] (11,0)--(11,8); 
\draw (13,6.4)--(14.5,7.6); 
\draw (16,6.4)--(14.5,7.6); 
\draw (13,5.6)--(13,4.4); 
\draw (16,5.6)--(16,4.4); 
\draw (13,2.4)--(13,3.6); 
\draw (16,2.4)--(16,3.6); 
\draw (13,1.6)--(14.5,0.4); 
\draw (16,1.6)--(14.5,0.4); 
\end{tikzpicture}}
\end{center}
\caption{Comparison between hard instances: an illustrating example when $d=2$.}
\label{fig:hardness-for-Ranking(d=2)}
\end{figure}
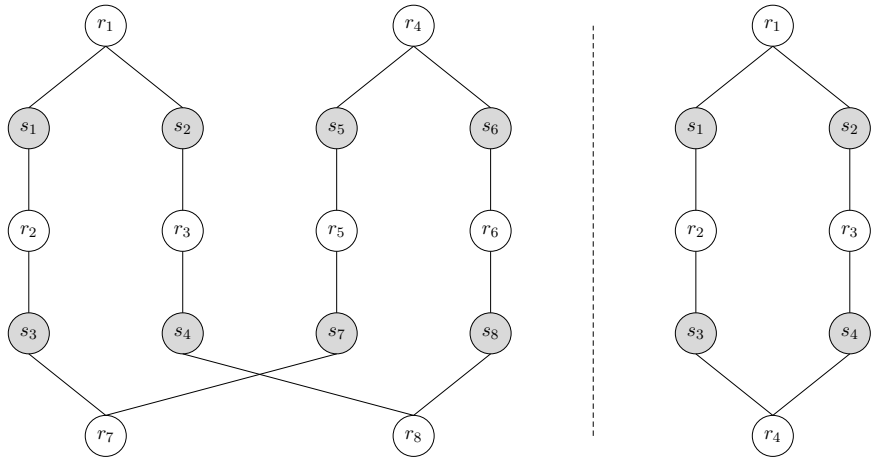

\section{OCS-based Algorithm} \label{sec:Discrete-OCS}

In this section, we propose algorithms for the online bipartite matching on $(d,d)$-bounded graphs based on Online Correlated Selection (OCS).
Recall that in the general OCS algorithms, when each online request $r$ arrives, there is an associated vector $\mathbf{x}^r: S \to [0,1]$ that assigns a weight $\mathbf{x}^r(s)$ to each server $s$, and the total weight $\sum_{s\in S} \mathbf{x}^r(s) \leq1$.
The online algorithm then picks an unmatched server with non-zero weight (if any), based on the vector $\mathbf{x}^r$ and the accumulated weight mass on all the servers.
Applying the OCS algorithm to the online bipartite matching on $(d,d)$-bounded graphs (by setting $\mathbf{x}^r(s) = 1/d$ if server $s$ is a neighbor of request $r$; $\mathbf{x}^r(s) = 0$ otherwise) yields the following algorithm (see Algorithm~\ref{alg:discrete-ocs}).
Suppose that the online requests are indexed by $1,2,\ldots,n$ (following their arrival order).

\begin{algorithm}[htb] \label{alg:discrete-ocs}
\caption{OCS-based Algorithm for OBM on $(d,d)$-bounded Graphs}
\textbf{Input:} Integer $d\geq 2$ \;
decide a non-decreasing sequence $1=f(0)\leq f(1)\leq \cdots \leq f(d)$ \;
\For{each online request $r\in R$}{
    \For{each server $s\in S$}{
        let $l_s^{r-1}$ be the number of neighbors of $s$ that have arrived (before $r$) \;
        let $U_s^{r-1} = 1$ if $s$ is unmatched; otherwise let $U_s^{r-1} = 0$ \;
    }
    let $A = \{ s\in N(r): U_s^{r-1} = 1 \}$ be the set of unmatched neighbors of $r$ \;
    \If{$A\neq \emptyset$}
        {match $r$ to $s$ chosen with probability $\frac{f(l_s^{r-1})}{\sum_{s'\in A} f(l_{s'}^{r-1})}$ \;}
}
\end{algorithm}

Note that when $d=2$ the above algorithm corresponds to the two-way semi-OCS~\cite{conf/focs/GaoHHNYZ21}, which has a competitive ratio $0.875$ by setting $f(1)$ to be an arbitrarily large number.

Throughout this section, we use $l_s^r$ to denote the degree of server $s$ at the end of the round when request $r$ arrives and let $U_s^r\in \{0,1\}$ be the indicator of $s$ being unmatched at the end of the round when request $r$ arrives.

\subsection{Characterizing the Competitive Ratio}
\label{ssec:characterize-candidate}

In this section, we give a lower bound on the competitive ratio of Algorithm~\ref{alg:discrete-ocs}.
Following the framework of Gao et al~\cite{conf/focs/GaoHHNYZ21}, we show that as long as the fixed sequence $f(0),f(1),\ldots,f(d)$ satisfies some constraints, we can guarantee that at the end of the round when a server $s$ reaches degree $i$, the probability that $s$ is unmatched is at most $1/f(i)$.
Specifically, we focus on the class of functions satisfying the following constraint and call them \emph{candidate functions}\footnote{While our algorithm only uses $f(0),f(1),\ldots,f(d-1)$, for generality we let the domain of a candidate function to be the set of all non-negative integers.}.
Note that the definition of candidate functions is with respect to a fixed $d\geq 2$.
For different $d$, the sets of candidate functions are different.

\begin{definition} [Candidate Function]
    Given any $d\geq 2$, we call a non-decreasing function $f: \{0,1,\ldots,\infty\} \to \mathbb{R}$ a candidate function if $f(0)=1$ and for any integers $0\leq m\leq d-1$ and $0\leq l_1 \leq l_2\leq \cdots \leq l_m < \infty$, it holds that
    \begin{equation} \label{equation:property-of-f}
        1 + \frac{\sum_{i=1}^m f(l_i)}{d-m} \geq \prod_{i=1}^m \frac{f(l_i+1)}{f(l_i)}.
    \end{equation}
    We use $\mathcal{C}_d$ to denote the set of all candidate functions.
\end{definition}

Roughly speaking, Equation~\eqref{equation:property-of-f} puts upper bounds on $f(l+1)$ in terms of $f(0),f(1),\ldots,f(l)$, for every $l\geq 1$, and thus prevents the function value from growing too fast.

The following lemma can be regarded as a discretized version of~\cite[Theorem 9]{conf/focs/GaoHHNYZ21}, which states that as long as $f$ is a candidate function, we can give an upper bound on the probability that a group of servers is unmatched, in terms of the candidate function values.
The proof of the lemma is deferred to Appendix~\ref{sec:missing-proofs(OCS)}.

\begin{lemma}\label{lemma:per-server-guarantee}
    For any candidate function $f\in \mathcal{C}_d$, for any $A \subseteq S$ and any $r \in R$, we have
    \begin{equation*}
        \Pr[U_{A}^{r}=1] \leq 1/\prod_{s\in A} f(l_s^r),
    \end{equation*}
    where $U_{A}^r = \prod_{s\in A} U_s^r \in \{0,1\}$ is the indicator of whether all the servers in $A\subseteq S$ are unmatched at the end of round $r$.
\end{lemma}

Particularly, by setting $A = \{s\}$ and $r$ as the last online request, Lemma~\ref{lemma:per-server-guarantee} indicates that for any server $s$, the probability that $s$ is unmatched when the algorithm terminates is at most $1/f(d)$ (since $s$ has at least $d$ neighbors).
Moreover, since the guarantee is per-vertex, we have a competitive ratio guarantee even if each offline server is associated with a non-negative weight.

\begin{corollary} \label{corollary:C.R.-of-ocs}
    Algorithm~\ref{alg:discrete-ocs} (with candidate function $f$) achieves a competitive ratio of $1-1/f(d)$ for the (vertex-weighted) online bipartite matching on $(d,d)$-bounded graphs.
\end{corollary}

Therefore, we translate the online matching problem into an optimization problem of maximizing $f(d)$ subject to $f\in \mathcal{C}_d$, i.e., $f$ is a candidate function.
While Constraint~\eqref{equation:property-of-f} looks complicated and puts lots of restrictions on $f$, it is not hard to find natural functions that are feasible.
For example, the function with $f(l)=1$ for all $l\geq 0$ is trivially feasible as the LHS of~\eqref{equation:property-of-f} is always at least $1$ while the RHS is always $1$, for any $0\leq l_1 \leq l_2\leq \cdots \leq l_m < \infty$.
It can also be verified that the function with $f(l)=(\frac{d}{d-1})^l$ is feasible since
\begin{equation*}
    1 + \frac{\sum_{i=1}^m f(l_i)}{d-m} \geq 1+\frac{m}{d-m} = \frac{d}{d-m} \geq (\frac{d}{d-1})^m = \prod_{i=1}^m \frac{f(l_i+1)}{f(l_i)}.
\end{equation*}

Using the above candidate function, for $(k,d)$-bounded graphs we can guarantee that each server is matched with probability at least $1-(1-1/d)^k$ when the algorithm terminates, which yields a competitive ratio of $1-(1-1/d)^k$, matching the optimal ratio for deterministic algorithms by~\cite{journals/teco/NaorW18}.
Another candidate function with $f(l) = \exp(\frac{l}{d}+\frac{l^2}{2d^2}+\frac{0.179 l^3}{d^3})$ is proposed in~\cite{conf/focs/GaoHHNYZ21}, which yields a competitive ratio of $1-\exp(-1-0.5-0.179)\approx 0.8134$.

Our main contribution in this section, is showing that there exists a polynomial-time algorithm that for every $d\geq 2$ computes the candidate function $f^*\in \mathcal{C}_d$ with maximum $f^*(d)$, which we refer to as the \emph{optimal candidate function}.
In fact, we impose a stronger definition for the optimal candidate function as follows.

\begin{definition} [Optimal Candidate Function]
    We call $f^* \in \mathcal{C}_d$ an optimal candidate function if for any $f\in \mathcal{C}_d$ and any integer $l\geq 0$, we have $f^*(l) \geq f(l)$.
\end{definition}


\begin{theorem} \label{theorem:optimal-candidate}
    The optimal candidate function in $\mathcal{C}_d$ can be computed in $O(d^2)$ time.
\end{theorem}

In the following section, we prove Theorem~\ref{theorem:optimal-candidate}.

\subsection{Computation of Optimal Candidate Function} \label{sec:properties-of-f}

In this section, we prove Theorem~\ref{theorem:optimal-candidate} by showing that the following function is the optimal candidate function in $\mathcal{C}_d$.

\begin{definition} \label{definition:optimal-candidate}
    Let $f^*$ be a function defined as follows: $f^*(0)=1$; for all $l\geq 1$, let
   \begin{equation*}
        f^*(l) := f^*(l-1) \cdot \min_{1\leq m< d} \left\{ \left(1+\frac{mf^*(l-1)}{d-m}\right)^{\frac{1}{m}} \right\}.
    \end{equation*}
\end{definition}

By definition, $f$ is an increasing function with $f(0)=1$.
Therefore to show that $f$ is a candidate function, it suffices to argue that it satisfies Constraint~\eqref{equation:property-of-f}.

The following technical lemma is crucial for showing the feasibility of $f^*$.

\begin{lemma} \label{lemma:technical-min-f}
    For every $m\in \{1,2,\ldots,d-1\}$, given any increasing sequence $1=a(0)<a(1)<\cdots<a(z)$ satisfying
    \begin{equation}
        1+\frac{m\cdot a(i)}{d-m} \geq \left( \frac{a(i+1)}{a(i)} \right)^m, \qquad \forall i\in\{0,1\ldots,z-1\}, \label{equation:condition-of-technical-lemma}
    \end{equation}
    we have
    \begin{equation}
        \min_{\substack{t\leq m,\\ \mathcal{L}\in [z-1]^{m-t}}} \left\{ \left( \left(1+\frac{t\cdot a(z)+\sum_{i=1}^{m-t}a(l_i)}{d-m}\right) \cdot \prod_{i=1}^{m-t} \frac{a(l_i)}{a(l_i+1)} \right)^{\frac{1}{t}} \right\}
        =
        \left(1+\frac{m\cdot a(z)}{d-m}\right)^{\frac{1}{m}},
        \label{equation:conclusion-of-technical-lemma}
    \end{equation}
    where $\mathcal{L}$ represents a sequence $(l_1,\ldots,l_{m-t})$ of integers with $0\leq l_1\leq \cdots \leq l_{m-t}\leq z-1$.
\end{lemma}

For continuity of presentation, we defer the proof of the above lemma to the end of this section.

\begin{lemma} \label{lemma:f*-feasible}
    Function $f^*$ is a candidate function.
\end{lemma}
\begin{proof}
    By definition, $f^*(0)=1$ and $f^*$ is increasing. 
    Thus it suffices to argue that $f^*$ satisfies Constraint~\eqref{equation:property-of-f}.
    Consider an arbitrarily fixed $m \in \{1,\ldots,d-1\}$.
    Note that each Constraint~\eqref{equation:property-of-f} is associated with a sequence $\mathcal{L} = (l_1,\ldots,l_m)$ with $0\leq l_1\leq \cdots \leq l_m$.
    We call a constraint of level $z$ if the associated sequence has $l_m = z$.
    We prove by mathematical induction on $z$ that $f^*$ satisfies all constraints of level $z$ (with the fixed $m$).

    By definition of $f^*$ we have $f^*(1) \leq \left( 1+\frac{m}{d-m} \right)^{\frac{1}{m}}$.
    Reordering the inequality gives
    \begin{equation*}
        1+\frac{m}{d-m} \geq (f^*(1))^m,
    \end{equation*}
    which means that the base case of $z=0$ holds.

    Now suppose that $f^*$ satisfies all constraints of levels $0,1,\ldots,z-1$, and we consider constraints of level $z$, where $z\geq 1$.
    Consider any sequence $\mathcal{L} = (l_1,\ldots,l_m)$ with $0\leq l_1\leq \cdots \leq l_m = z$, and suppose $z$ appears $t\geq 1$ times in the sequence, i.e., we have $l_m=l_{m-1}=\cdots=l_{m-t+1}=z$ and $l_{m-t}\leq z-1$.
    Since $f^*$ satisfies all constraints of level at most $z-1$, for all $i\in\{0,1\ldots,z-1\}$, using $\mathcal{L}' = (i,i,\ldots,i)$ in Constraint~\eqref{equation:property-of-f},
    we have
    \begin{equation*}
        1+\frac{m\cdot f^*(i)}{d-m} \geq \left( \frac{f^*(i+1)}{f^*(i)} \right)^m.
    \end{equation*}

    Hence $f^*$ satisfies Condition~\eqref{equation:condition-of-technical-lemma}.
    By Lemma~\ref{lemma:technical-min-f} and definition of $f^*$, we have
    \begin{align*}
        f^*(z+1) &\leq f^*(z) \cdot \left(1+\frac{m\cdot f^*(z)}{d-m}\right)^{\frac{1}{m}} \\
        &\leq f^*(z) \cdot \left( \left(1+\frac{t\cdot f^*(z)+\sum_{i=1}^{m-t} f^*(l_i)}{d-m}\right) \cdot \prod_{i=1}^{m-t} \frac{f^*(l_i)}{f^*(l_i+1)} \right)^{\frac{1}{t}} \\
        &\leq f^*(z) \cdot \left( \left(1+\frac{\sum_{i=1}^{m} f^*(l_i)}{d-m}\right) \cdot \prod_{i=1}^{m-t} \frac{f^*(l_i)}{f^*(l_i+1)} \right)^{\frac{1}{t}}.
    \end{align*}
    
    Reordering the above inequality, we have
    \begin{equation*}
        1 + \frac{\sum_{i=1}^{m} f^*(l_i)}{d-m} \geq \left(\frac{f^*(z+1)}{f^*(z)}\right)^t \cdot \prod_{i=1}^{m-t} \frac{f^*(l_i+1)}{f^*(l_i)} = \prod_{i=1}^{m} \frac{f^*(l_i+1)}{f^*(l_i)}.
    \end{equation*}

    In other words, $f^*$ satisfies Constraint~\eqref{equation:property-of-f} with $m$ and $\mathcal{L}$.
    Since $m$ and $\mathcal{L}$ are fixed arbitrarily, we conclude that $f^*$ is a candidate function.
\end{proof}

Next, we show that $f$ is an optimal candidate function.

\begin{proofof}{Theorem~\ref{theorem:optimal-candidate}}
    We prove the theorem by showing that $f^*$ is an optimal candidate function.
    Specifically, we show that for any $l\in \mathbb{N}$ and any function $h\in \mathcal{C}_d$, we have $f^*(l) \geq h(l)$.
    We prove the statement by mathematical induction on $l = 0,1,\ldots$.
    
    When $l=0$, the statement holds trivially since $f^*(0) = h(0) = 1$.
    
    Next, we assume that $f^*(l-1) \geq h(l-1)$, and show that we have $f^*(l) \geq h(l)$, finishing the inductive step.
    Since $h$ is a candidate function, it must satisfy Constraint~\eqref{equation:property-of-f}.
    In particular, setting $l_1 = l_2 = \cdots l_m = l-1$ for every $m\in\{1,2,\ldots,d-1\}$, we have
    \begin{align*}
        h(l) & \leq h(l-1) \cdot \min_{1\leq m<d} \left\{ \left(1+\frac{m \cdot h(l-1)}{d-m}\right)^{\frac{1}{m}} \right\} \\
        & \leq f^*(l-1) \cdot \min_{1\leq m<d} \left\{ \left(1+\frac{m \cdot f^*(l-1)}{d-m}\right)^{\frac{1}{m}} \right\} = f^*(l),
    \end{align*}
    where the second inequality holds by induction hypothesis.
    
    Finally, by definition of $f^*$, given the value of $f^*(l-1)$ we can compute $f^*(l)$ in $O(d)$ time. Therefore the function values $f^*(0),f^*(1),\ldots,f^*(d)$ can be computed in $O(d^2)$ time.
\end{proofof}

\subsection{The Competitive Ratio}

By the above argument, for every $d\geq 2$, it suffices to compute the optimal candidate function (in $O(d^2)$ time), and use it in Algorithm~\ref{alg:discrete-ocs}, which guarantees a competitive ratio of $1-1/f^*(d)$ for the online bipartite matching on $(d,d)$-bounded graphs.
In the following, we use $f_d \in \mathcal{C}_d$ to denote the optimal candidate function for $(d,d)$-bounded graphs.
We list the function values for $d\in\{3,4,\ldots,10\}$ in the following table for an illustrative presentation.
The corresponding (lower bounds of) competitive ratios are shown below.
We also list the competitive ratio for higher values of $d$ as examples.
Note that for $d\leq 3300$, the competitive ratio of our algorithm is always larger than that of~\cite{conf/soda/CohenW18}.

\begin{table}[htb]
    \begin{center}
        \begin{tabular}{ c || c|c|c|c|c|c|c|c|c|c } 
            $d \backslash l$ & $1$ & $2$ & $3$ & $4$ & $5$ & $6$ & $7$ & $8$ & $9$ & $10$ \\
            \hline 
            \hline 
            $3$ & $1.5$ & $2.625$ & $6.0703$ & - & - & - & - & - & - & - \\ 
            $4$ & $1.3333$ & $1.9259$ & $3.1623$ & $6.4516$ & - & - & - & - & - & - \\ 
            $5$ & $1.25$ & $1.6406$ & $2.3135$ & $3.6516$ & $6.7673$ & - & - & - & - & - \\ 
            $6$ & $1.2$ & $1.488$ & $1.9308$ & $2.6764$ & $4.0926$ & $7.0412$ & - & - & - & - \\
            $7$ & $1.1666$ & $1.3935$ & $1.7171$ & $2.2086$ & $3.0216$ & $4.4831$ & $7.2863$ & - & - & - \\
            $8$ & $1.1428$ & $1.3294$ & $1.5819$ & $1.9394$ & $2.4767$ & $3.3464$ & $4.8307$ & $7.5050$ & - & - \\
            $9$ & $1.125$ & $1.2832$ & $1.4890$ & $1.7661$ & $2.1561$ & $2.7372$ & $3.6486$ & $5.1336$ & $7.6643$ & - \\
            $10$ & $1.1111$ & $1.2482$ & $1.4214$ & $1.6459$ & $1.9469$ & $2.3680$ & $2.9879$ & $3.9297$ & $5.4065$ & $7.8134$ \\
            \end{tabular}
        \end{center}
    \caption{Values of $f_d(l)$ for different $d$ and $l$.}
    \label{fig:values-for-f}
\end{table}

\begin{table}[htb]
    \begin{center}
        \begin{tabular}{|c|c|c|c|c|c|c|c|c|c} 
           \hline
           $d$ & $3$ & $4$ & $5$ & $6$ & $7$ & $8$ & $9$ & $10$ \\
           \hline
           ratio & $0.8352$ & $0.8450$ & $0.8522$ & $0.8579$ & $0.8627$ & $0.8667$ & $0.8695$ & $0.8720$  \\
           \hline
        \end{tabular}
    \end{center}
    \caption{Lower bound on the competitive ratio of OCS-based algorithm with $d\in \{3,4,\ldots,10\}$.}
    \label{fig:lower-bound-OCS-small-d}
\end{table}

\begin{table}[htb]
    \begin{center}
        \begin{tabular}{|c|c|c|c|c|c|c|c|c|c|c} 
           \hline
           $d$ & $20$ & $40$ & $80$ & $200$ & $400$ & $800$ & $2000$ & $4000$ & $8000$ \\
           \hline
           ratio & $0.8842$ & $0.8907$ & $0.8941$ & $0.8962$ & $0.8969$ & $0.8972$ & $0.8974$ & $0.8975$ & $0.8976$  \\
           \hline
        \end{tabular}
    \end{center}
    \caption{Lower bound on the competitive ratio of OCS-based algorithm with larger $d$.}
    \label{fig:lower-bound-OCS-large-d}
\end{table}

Since computing $f_d(d)$ when $d$ is very large is impractical, in the following, we establish a lower bound on the competitive ratio of Algorithm~\ref{alg:discrete-ocs} when $d\to \infty$.

\begin{theorem} \label{theorem:C.R.-of-OCS-large-d}
    Algorithm~\ref{alg:discrete-ocs} is at least $0.8976$-competitive for online bipartite matching on $(d,d)$-bounded graphs when $d\to \infty$.
\end{theorem}

We prove Theorem~\ref{theorem:C.R.-of-OCS-large-d} by defining a function $g_{d}(l)$ that always serves as a lower bound of $f_{d}(l)$ for any $l\in \mathbb{N}$ and giving a lower bound for $g_d(d)$.

\begin{definition} \label{definition:lower-bound-of-f*}
    For any fixed $d$, let $g_{d}(l)$ be defined as follows: $g_{d}(0)=1$; for all $l\geq 1$, let
    \begin{equation*}
        g_{d}(l) := g_{d}(l-1)\cdot \min_{t\in (0,1)} \left\{ \left( \frac{1+t\cdot g_{d}(l-1)}{1-t} \right)^\frac{1}{t\cdot d} \right\}.
    \end{equation*}
\end{definition}

\begin{lemma}\label{lemma:lower-bound-f-by-g}
 For all $d\geq 2$ and $l\geq 0$, we have $f_d(l)\geq g_d(l)$.  
\end{lemma}
\begin{proof}
    Fix an arbitrary $d\geq 2$. We prove the statement by induction on $l$.
    Since $f_d(0) = g_d(0) = 1$, the base case of $l=0$ holds trivially.
    Now suppose that we have $f_d(l-1)\geq g_d(l-1)$. By definition of $f_d$ and $g_d$, we have
    \begin{align*}
        f_d(l) & = f_d(l-1)\cdot \min_{1\leq m< d} \left\{ \left(1+\frac{m\cdot f_d(l-1)}{d-m}\right)^{\frac{1}{m}} \right\} \\
        & \geq g_d(l-1)\cdot \min_{1\leq m< d} \left\{ \left(1+\frac{m\cdot g_d(l-1)}{d-m}\right)^{\frac{1}{m}} \right\} \\
        & = g_d(l-1)\cdot \min_{t \in \{\frac{1}{d}, \ldots, \frac{d-1}{d}\}} \left\{ \left(1+\frac{t\cdot g_d(l-1)}{1-t}\right)^{\frac{1}{t\cdot d}} \right\} \\
        & \geq g_d(l-1)\cdot \min_{t \in (0,1)} \left\{ \left(1+\frac{t\cdot g_d(l-1)}{1-t}\right)^{\frac{1}{t\cdot d}} \right\} = g_d(l),
    \end{align*}
    where the first inequality holds by induction hypothesis.
    By mathematical induction, the statement is true for all $l\geq 0$ and the lemma follows.
\end{proof}

In the following, we give a lower bound for $g_d(d)$ when $d$ is sufficiently large.

\begin{lemma} \label{lemma:property-of-g}
    For any fixed $d$ and $l$, we have $g_{2d}(2l) \geq g_{d}(l)$.
\end{lemma}
\begin{proof}
    Fix an arbitrary $d$.
    We prove the statement by mathematical induction on $l\geq 0$.
    When $l=0$, it holds trivially since $g_{2d}(0)=g_{d}(0)$.
    
    Next, we assume that $g_{2d}(2(l-1))\geq g_{d}(l-1)$, and show that $g_{2d}(2 l)\geq g_{d}(l)$.
    We have
    \begin{align*}
        & \ g_{2d}(2 l) = g_{2d}(2 l-1)\cdot \min_{t\in (0,1)} \left\{ \left( \frac{1+t\cdot g_{2d}(2 l-1)}{1-t} \right)^\frac{1}{t\cdot 2d} \right\}\\
        = & \ g_{2d}(2l-2))\cdot \left( \min_{t\in (0,1)} \left\{ \left(\frac{1+t\cdot g_{2d}(2l-2))}{1-t}\right)^\frac{1}{t\cdot 2 d} \right\} \right) \cdot \left(\min_{t\in (0,1)} \left\{ \left(\frac{1+t\cdot g_{2d}(2l-1)}{1-t}\right)^\frac{1}{t\cdot 2 d} \right\} \right) \\
        \geq & \ g_{2d}(2(l-1))\cdot \left( \min_{t\in (0,1)} \left\{ \left(\frac{1+t\cdot g_{2d}(2(l-1)))}{1-t}\right)^\frac{1}{t\cdot d} \right\} \right) \\
        \geq & \ g_{d}(l-1)\cdot \min_{t\in (0,1)} \left\{ \left( \frac{1+t\cdot g_{d}(l-1)}{1-t} \right)^\frac{1}{t\cdot d} \right\} = g_{d}(l),
    \end{align*}
    where the first inequality holds by monotonicity of function $g_{2d}$ and the second inequality holds due to the induction hypothesis.
    By mathematical induction, the statement holds for all $l\geq 0$ and the lemma follows.
\end{proof}

Equipped with the above lemma, to give a lower bound for $g_d(d)$ when $d\to \infty$, it suffices to lower bound $g_d(d)$ when $d$ is a large constant.
Hence the following claim (whose proof can be found in Appendix~\ref{sec:missing-proofs(OCS)}) implies Theorem~\ref{theorem:C.R.-of-OCS-large-d} because the competitive ratio is lower bounded by
\begin{equation*}
    1-\frac{1}{f_d(d)}\geq 1-\frac{1}{g_d(d)} \geq 1-\frac{1}{g_{10000}(10000)} \geq 1-\frac{1}{9.7657} \geq 0.8976.
\end{equation*}

\begin{claim} \label{claim:g_10000(10000)}
    We have $g_{10000}(10000) \geq 9.7657$.
\end{claim}

\subsection{Proof of the Technical Lemma}

It remains to prove the technical lemma (Lemma~\ref{lemma:technical-min-f}).

\begin{proofof}{Lemma~\ref{lemma:technical-min-f}}
    Consider any fixed $m\in \{1,2,\ldots,d-1\}$.
    We prove the lemma by induction on $z$.
    The case of $z=0$ is trivial as we have $t=m$ and $\mathcal{L}=\emptyset$.
    Now suppose that the statement is true for $z-1$, and we consider the case of $z$.
    
    Note that the RHS of Equation~\eqref{equation:conclusion-of-technical-lemma} in Lemma~\ref{lemma:technical-min-f} can be obtained by setting $t=m$ (thus $\mathcal{L}$ is an empty sequence) in the LHS.
    Therefore, to prove the lemma it suffices to show that for $t < m$ and sequence $0\leq l_1\leq \cdots \leq l_{m-t}\leq z-1$, the term being minimized on the LHS is always at least the RHS.
    For ease of notation, we let
    \begin{equation*}
        A(z,t,\mathcal{L}) = \left( \left(1+\frac{t\cdot a(z)+\sum_{i=1}^{m-t}a(l_i)}{d-m}\right) \cdot \prod_{i=1}^{m-t} \frac{a(l_i)}{a(l_i+1)} \right)^{\frac{1}{t}} \text{ and }
        B(z) = \left(1+\frac{m\cdot a(z)}{d-m}\right)^{\frac{1}{m}}.
    \end{equation*}

    Note that $B(z) = A(z,m,\emptyset)$.
    Also note that by induction hypothesis (for $z-1$), we have 
    \begin{equation*}
        A(z-1, t,\mathcal{L}) \geq B(z-1), \qquad \forall t \leq m, \mathcal{L}\in [z-2]^{m-t}.
    \end{equation*}

    Fixed any $t<m$ and $0\leq l_1\leq \cdots \leq l_{m-t}\leq z-1$, we show that $A(z, t,\mathcal{L}) \geq B(z)$.
    Equivalently, we show that $\left( \frac{A(z,t,\mathcal{L})}{B(z)} \right)^{mt}\geq 1$.
    Note that (by definition of $A(z,t,\mathcal{L})$ and $B(z)$)
    \begin{equation*}
        \left( \frac{A(z,t,\mathcal{L})}{B(z)} \right)^{mt} = \left( \prod_{i=1}^{m-t} \frac{a(l_i)}{a(l_i+1)} \right)^{m}\frac{\left(d-m+t\cdot a(z)+\sum_{i=1}^{m-t}a(l_i) \right)^m}{\left(d-m+m\cdot a(z)\right)^t\cdot (d-m)^{m-t}}.
    \end{equation*}

    Define function (where $x > 0$)
    \begin{equation*}
        b(x) := \frac{\left(d-m + t\cdot x + \sum_{i=1}^{m-t} a(l_i) \right)^m}{(d-m + m\cdot x)^t}.
    \end{equation*}

    For ease of presentation, we use $\diamondsuit$ and $\heartsuit$ to denote $d-m + t\cdot x + \sum_{i=1}^{m-t} a(l_i)$ and $d-m + m\cdot x$, respectively.
    Note that we have $\diamondsuit > 0$ and $\heartsuit > 0$.
    Take the derivative of $b(x)$, we have
    \begin{align*}
        b'(x) &= \frac{(m \cdot \diamondsuit^{m-1} \cdot t) \cdot \heartsuit^t - \diamondsuit^m \cdot (t \cdot \heartsuit^{t-1} \cdot m)}{\heartsuit^{2t}} \\
        &= \frac{(mt \cdot \diamondsuit^{m-1} \cdot \heartsuit^{t-1}) \cdot (\heartsuit-\diamondsuit)}{\heartsuit^{2t}} = \frac{(mt \cdot \diamondsuit^{m-1} \cdot \heartsuit^{t-1})}{\heartsuit^{2t}} \cdot \left((m-t)x - \sum_{i=1}^{m-t} a(l_i)\right).
    \end{align*}
    
    Therefore when $(m-t)x \geq \sum_{i=1}^{m-t} a(l_i)$, $b(x)$ is non-decreasing.
    Since $a(z) > a(z-1) \geq a(l_i)$ for all $i\in \{1,2,\ldots,m-t\}$, by monotonicity of function $b$, we have
    \begin{align}
        \left( \frac{A(z,t,\mathcal{L})}{B(z)} \right)^{mt} & = \left( \prod_{i=1}^{m-t} \frac{a(l_i)}{a(l_i+1)} \right)^{m}\frac{ b(a(z))}{ (d-m)^{m-t}} \geq \left( \prod_{i=1}^{m-t} \frac{a(l_i)}{a(l_i+1)} \right)^{m}\frac{ b(a(z-1))}{ (d-m)^{m-t}} \nonumber \\
        & = \left( \prod_{i=1}^{m-t} \frac{a(l_i)}{a(l_i+1)} \right)^{m}\frac{\left(d-m+t\cdot a(z-1)+\sum_{i=1}^{m-t}a(l_i) \right)^m}{\left(d-m+m\cdot a(z-1)\right)^t\cdot (d-m)^{m-t}} \label{equation:RHS-of-(A/B)^mt}.
    \end{align}

    Recall that $0\leq l_1\leq \cdots \leq l_{m-t}\leq z-1$.
    Let $t'\geq 0$ be the number of appearances of $z-1$ in the sequence, i.e., we have $l_{m-t} = l_{m-t-1} = \cdots = l_{m-t-t'+1} = z-1$ and $l_{m-t-t'} \leq z-2$.
    Then we can write the RHS of~\eqref{equation:RHS-of-(A/B)^mt} as
    \begin{equation}
        \left( \prod_{i=1}^{m-t-t'} \frac{a(l_i)}{a(l_i+1)}\right)^{m} \left(\frac{a(z-1)}{a(z)} \right)^{mt'} \frac{\left(d-m+(t+t')\cdot a(z-1)+\sum_{i=1}^{m-t-t'}a(l_i) \right)^m}{\left(d-m+m\cdot a(z-1)\right)^t\cdot (d-m)^{m-t}}. \label{equation:new-RHS-of-(A/B)^mt}
    \end{equation}

    Recall that by Condition~\eqref{equation:condition-of-technical-lemma}, we have
    \begin{equation*}
        \left(\frac{a(z)}{a(z-1)} \right)^{m} \leq 1+\frac{m\cdot a(z-1)}{d-m} = \frac{d-m+m\cdot a(z-1)}{d-m}.
    \end{equation*}
    
    Plugging in this bound to Equation~\eqref{equation:new-RHS-of-(A/B)^mt}, we can lower bound $\left( \frac{A(z,t,\mathcal{L})}{B(z)} \right)^{mt}$ as
    \begin{equation*}
        \left( \frac{A(z,t,\mathcal{L})}{B(z)} \right)^{mt} \geq
        \left( \prod_{i=1}^{m-t-t'} \frac{a(l_i)}{a(l_i+1)}\right)^{m} \frac{\left(d-m+(t+t')\cdot a(z-1)+\sum_{i=1}^{m-t-t'}a(l_i) \right)^m}{\left(d-m+m\cdot a(z-1)\right)^{t+t'}\cdot (d-m)^{m-t-t'}}.
    \end{equation*}

    Note that in the above equation, we have $t+t'\leq m$ and $0\leq l_1\leq \cdots \leq l_{m-t-t'}\leq z-2$.
    Let $\mathcal{L}' = (l_1,\ldots,l_{m-t-t'})$ be the length-$(m-t-t')$ prefix of $\mathcal{L}$, the RHS of the above equation is exactly $\left( \frac{A(z-1,t+t',\mathcal{L'})}{B(z-1)} \right)^{mt}$.
    Therefore we have
    \begin{equation*}
        \left( \frac{A(z,t,\mathcal{L})}{B(z)} \right)^{mt} \geq \left( \frac{A(z-1,t+t',\mathcal{L'})}{B(z-1)} \right)^{mt} \geq 1,
    \end{equation*}
    where the last inequality holds by induction hypothesis.
\end{proofof}

\section{Conclusion}

In this paper, we consider the online bipartite matching problem on degree-bounded graphs and analyze the competitive ratios of \Ranking and OCS.
We establish an upper bound on the competitive ratio of \Ranking for every $d\geq 2$, and show that it converges to $0.816$ as $d$ grows.
We further show that the competitive ratio of an OCS-based algorithm is always at least $0.835$ (for all $d\geq 3$), and can be as good as $0.897$ when $d$ is large.
Our work shows that in the degree-bounded environment, OCS is a better algorithm than \Ranking.
Our work uncovers some interesting open problems.
First, while we show that the candidate function we use in the OCS-based algorithm is optimal (under the analysis framework of~\cite{conf/focs/GaoHHNYZ21}), the obtained ratio does not seem optimal (when $d\geq 3$).
It would be an interesting future work to design an OCS-based algorithm that achieves better competitive ratios, e.g., better than $0.875$ for all $d\geq 3$ and approaching $1$ when $d\to \infty$.
Another open problem is to investigate other assumptions weaker than the degree-bounded condition on the bipartite graphs that allow for a good competitive ratio, e.g., larger than $1-1/e$.

\section*{Acknowledgement}

Xiaowei Wu is funded by the Science and Technology Development Fund (FDCT), Macau SAR (file no. 0147/2024/RIA2, 0014/2022/AFJ, 0085/2022/A, 001/2024/SKL and CG2025-IOTSC).
The authors would like to thank Will Ma and Zhihao Tang for the helpful discussions on the results of the paper.


\newcommand{\etalchar}[1]{$^{#1}$}

\newpage

\appendix

\section{Limitations of the Marking Algorithm}
\label{sec:hardness-for-marking}

In this section, we briefly discuss the \textsc{Marking} algorithm proposed in \cite{conf/soda/CohenW18} and show its limitations when $d$ is small.
Unlike the algorithms we consider in this paper, the \textsc{Marking} algorithm will sometimes \emph{mark} an unmatched server and pretend that it is matched (for the purpose of controlling the variance of the matching outcomes, intuitively speaking).
A server will also be marked when it is matched.
Upon the arrival of each request $r$, \textsc{Marking} defines a weight $w(s)$ for every offline server $s$ depending on its current degree, where a larger degree implies a larger weight.
Then, the online request $r$ will be matched with probability
\begin{equation*}
    P_1(r) = \min\left\{ \frac{w(C)}{d},1 \right\}\cdot (1-\epsilon),
\end{equation*}
where an unmarked neighboring server will be chosen with a probability proportional to its weight. 
Here, $w(C)$ is the total weight of unmarked neighbors of $r$ upon $r$'s arrival, and $\epsilon$ is a parameter of the algorithm (which is set to be $1/\sqrt{d}$ in their paper).
Afterward, the algorithm marks each unmarked neighbor $s$ of $r$ with a certain probability $P_2(s)$ depending on $w(s)$.
By carefully defining the weight $w(s)$ and the marking probability $P_2(s)$ for every unmarked server $s$, it can be proved (see Lemma 4.10 of \cite{conf/soda/CohenW18}) that $P_1\approx 1$ and $P_2\approx 0$ when $d$ is sufficiently large.
Roughly speaking, for large $d$, we have $\epsilon = 1/\sqrt{d} \approx 0$ and it can be shown that $w(C)$ is highly concentrated around its expectation $d$.
Therefore the \textsc{Marking} algorithm performs very well when $d = \omega(1)$.
However, for small $d$ the performance of the algorithm could be sub-optimal.

Specifically, the \textsc{Marking} algorithm suffers a loss in competitive ratio due to two main reasons.
First, some of the servers are artificially marked and thus cannot be matched throughout the whole algorithm while being unmatched.
Second, the probability of each online request being matched can be strictly smaller than $1$, even when it has unmarked neighbors: this happens when $\epsilon>0$ or when $w(C)$ is (significantly) smaller than $d$.
Note that having $P_1(r) < 1$ is crucial for upper bounding the variance of the matching outcomes in their algorithm and analysis.

Consider a more concrete example shown in Figure~\ref{fig:hardness-for-Marking(d=2)} with $n$ servers $\{s_1,s_2,\ldots,s_n\}$ and $n$ requests $\{r_1,r_2,\ldots,r_n\}$, where the requests arrive in the order of $r_1,r_2,\ldots,r_n$.

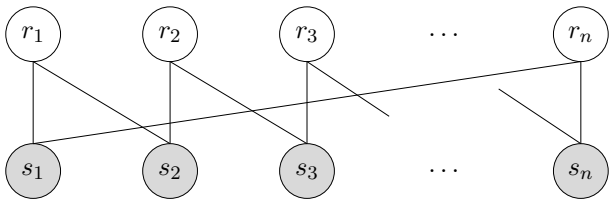
\begin{figure}[htb]
\begin{center}
\resizebox{0.5\textwidth}{!}{
\begin{tikzpicture}
\draw [fill = gray!30] (0,0) circle (0.4); \node at (0,0) {$s_1$};
\draw [fill = gray!30] (2,0) circle (0.4); \node at (2,0) {$s_2$};
\draw [fill = gray!30] (4,0) circle (0.4); \node at (4,0) {$s_3$};
\node at (6,0) {$\ldots$};
\draw [fill = gray!30] (8,0) circle (0.4); \node at (8,0) {$s_n$};
\draw (0,2) circle (0.4); \node at (0,2) {$r_1$};
\draw (2,2) circle (0.4); \node at (2,2) {$r_2$};
\draw (4,2) circle (0.4); \node at (4,2) {$r_3$};
\node at (6,2) {$\ldots$};
\draw (8,2) circle (0.4); \node at (8,2) {$r_n$};
\draw (0,0.4)--(0,1.6); 
\draw (2,0.4)--(0,1.6); 
\draw (2,0.4)--(2,1.6); 
\draw (4,0.4)--(2,1.6); 
\draw (4,0.4)--(4,1.6); 
\draw (5.2,0.8)--(4,1.6); 
\draw (8,0.4)--(6.8,1.2); 
\draw (8,0.4)--(8,1.6); 
\draw (0,0.4)--(8,1.6); 
\end{tikzpicture} }
\end{center}
\caption{Hard instance for \textsc{Marking} when $d=2$.
Requests are indexed by their arrival order, and each $r_i$ has two neighbors $s_i$ and $s_{i+1}$, where $s_{n+1} = s_1$.}
\label{fig:hardness-for-Marking(d=2)}
\end{figure}

For each $1<i<n$, it can be shown (see Lemma 4.3 of \cite{conf/soda/CohenW18}) that the probability of $s_i$ being unmarked when $r_i$ arrives is $(1+\epsilon)/2$.
Therefore, we have
\begin{equation*}
    P_1(r_i) =
    \begin{cases}
        \frac{1-\epsilon}{2}, & \text{w.p. } \frac{1-\epsilon}{2}\\
        1-\epsilon, & \text{w.p. } \frac{1+\epsilon}{2}.
    \end{cases}
\end{equation*}

As a consequence, $r_i$ gets matched with probability $(3-2\epsilon-\epsilon^2)/4 \leq 3/4$ for all $\epsilon\geq 0$ and the competitive ratio of \textsc{Marking} (for $d=2$) is at most $3/4$ (when $n\to \infty$).
Note that the $3/4$ ratio can be trivially achieved by the naive algorithm \textsc{Random} and is clearly sub-optimal.

\section{Missing Proofs in Section~\ref{sec:hardness-for-Ranking}} \label{sec:missing-proofs(Ranking)}

\begin{lemma} \label{lemma:stochastic-equiv}
    For any $(t_1,t_2,\ldots,t_{d-1})\in [0,1]^{d-1}$, we have
    \begin{equation} \label{equation:stochastic-equiv}
        \Pr[g_1<t_1,g_2<t_2,\cdots,g_{d-1}<t_{d-1}] = \Pr[\theta_1<t_1,\theta_2<t_2,\cdots,\theta_{d-1}<t_{d-1}].
    \end{equation}
\end{lemma}
\begin{proof}
    We prove the lemma by mathematical induction on $d\geq 2$.
    When $d=2$, the statement holds trivially since $g_1$ and $\theta_1$ are i.i.d..
    Next, we assume that the statement holds for $d-1$ and consider the case of $d$.
    More specifically, for any $(x_1,x_2\ldots,x_{d-2})\in [0,1]^{d-2}$, we have
    \begin{equation*}
        \Pr\left[\prod_{j=1}^{d-2} m_j<x_1, \prod_{j=2}^{d-2} m_j<x_1, \ldots, m_{d-2}<x_{d-2}\right] = \Pr[\xi_1<x_1,\xi_2<x_2,\ldots,\xi_{d-2}<x_{d-2}],
    \end{equation*}
    where $\xi_1\leq \xi_2\leq \ldots\leq \xi_{d-2}$ are $d-2$ sorted random variables drawn from $U[0,1]$ independently.
    
    Fix any $t=(t_1,\ldots,t_{d-1})\in [0,1]^{d-1}$.
    For the LHS of Equation~\eqref{equation:stochastic-equiv}, we have
    \begin{align*}
        \Pr[g_1<t_1,\ldots,g_{d-1}<t_{d-1}] & = \int_{0}^{t_{d-1}} \left(\Pr[g_1<t_1,\ldots,g_{d-2}<t_{d-2} \mid g_{d-1}=x] \cdot \Pr[g_{d-1}=x]\right) \mathrm{d}x \\
        & = \int_{0}^{t_{d-1}} \left(\Pr\left[\prod_{j=1}^{d-2} m_j\cdot x<t_1,\ldots,m_{d-2}\cdot x<t_{d-2}\right] \cdot \Pr[g_{d-1}=x]\right) \mathrm{d}x \\
        & = \int_{0}^{t_{d-1}} \left(\Pr[\xi_1\cdot x<t_1,\ldots,\xi_{d-2}\cdot x<t_{d-2}] \cdot \Pr[\theta_{d-1}=x]\right) \mathrm{d}x,
    \end{align*}
    where the second equality holds due to Definition~\ref{definition:m-and-g}, and the last equality holds due to the induction hypothesis.
    For the RHS of Equation~\eqref{equation:stochastic-equiv}, we have
    \begin{align*}
        \Pr[\theta_1<t_1,\ldots,\theta_{d-1}<t_{d-1}] & = \int_{0}^{t_{d-1}} \left(\Pr[\theta_1<t_1,\ldots,\theta_{d-2}<t_{d-2} \mid \theta_{d-1}=x] \cdot \Pr[\theta_{d-1}=x]\right) \mathrm{d}x.
    \end{align*}
    
    Observe that $\Pr[\theta_1<t_1,\ldots,\theta_{d-2}<t_{d-2} \mid \theta_{d-1}=x] = \Pr[\xi_1\cdot x<t_1,\ldots,\xi_{d-2}\cdot x<t_{d-2}]$ because conditioned on $\theta_{d-1} = x$, we can equivalently think of $(\theta_1, \theta_2,\ldots, \theta_{d-2})$ as $d-2$ sorted random variables drawn from $U[0,x]$.
    Hence we finish the proof.
\end{proof}

\begin{proofof}{Lemma~\ref{lemma:hardness-for-Ranking(d=2)}}
    Consider the hard instance given in Figure~\ref{fig:hardness-for-Ranking(d=2)} (left) where we assume requests arrive in the order of $r_1,r_2,\ldots$.
    Obviously there exists a perfect matching in the graph, and we have $\opt = 8$.
    Moreover, since requests $\{r_i\}_{i \in \{1,2,3,4,5,6\}}$ will always get matched by Ranking, it suffices to consider the matching status of $r_7$ and $r_8$.
    
    We focus on request $r_7$ first and consider the matching status of its two neighbors $s_3$ and $s_7$ when $r_7$ arrives.
    We have
    \begin{align*}
        \Pr[s_3 \text{ is unmatched when } r_7 \text{ arrives}] & = \Pr[r_1 \to s_2, r_2 \to s_1] \\
        & = \Pr[y_{s_2} < y_{s_1} < y_{s_3}] = 1/6.
    \end{align*}

    Symmetrically we have
    \begin{equation*}
        \Pr[s_7 \text{ is unmatched when } r_7 \text{ arrives}] = 1/6.
    \end{equation*}

    Moreover, the events of $s_3$ and $s_7$ being matched are independent.
    Thus we have
    \begin{equation*}
        \Pr[r_7 \text{ gets matched}] = 1 - \Pr[\text{both } s_3, s_7 \text{ are matched when } r_7 \text{ arrives}] = 1 - (5/6)^2 = 11/36.
    \end{equation*}

    Symmetrically, we can show that the probability that $r_8$ gets matched equals $11/36$ as well (however, the events of $r_7$ and $r_8$ being matched are not independent).
    Therefore, we have
    \begin{equation*}
        \frac{\E[\alg]}{\opt} = \frac{6 + 2 \cdot (11/36)}{8} = 119/144 \approx 0.8264
    \end{equation*}
    in this instance, and we finish the proof.
\end{proofof}

\section{Missing Proofs in Section~\ref{sec:Discrete-OCS}} \label{sec:missing-proofs(OCS)}

\begin{proofof}{Lemma~\ref{lemma:per-server-guarantee}}
    For any $A \subseteq S$ and any $r \in R$, the probability that all the servers in $A$ are unmatched at the end of round $r$ is
    \begin{equation*}
        \Pr[U_{A}^{r}=1] = \E_r[U_{A}^{r}] = \E_{r-1}\left[ U_{A}^{r-1} \cdot \frac{\sum_{s\in N(r) \setminus A} f(l_s^{r-1}) \cdot U_s^{r-1}}{\sum_{s\in N(r)} f(l_s^{r-1}) \cdot U_s^{r-1}} \right],
    \end{equation*}
    where the first expectation is taken over the randomness of the first $r$ rounds while the second is taken over the first $r-1$ rounds. 
    Multiply both numerator and denominator by $U_{A}^{r-1}$, we have
    \begin{align*}
        \Pr[U_{A}^{r}=1] & = \E_{r-1}\left[ U_{A}^{r-1} \cdot \frac{\sum_{s\in N(r)\setminus A} f(l_s^{r-1}) \cdot U_s^{r-1}}{\sum_{s\in N(r)\setminus A} f(l_s^{r-1}) \cdot U_s^{r-1} + \sum_{s\in A\cap N(r)} f(l_s^{r-1}) \cdot U_s^{r-1}} \right] \\
        & = \E_{r-1}\left[ U_{A}^{r-1} \cdot \frac{\sum_{s\in N(r)\setminus A} f(l_s^{r-1}) \cdot U_{A\cup \{s\}}^{r-1}}{\sum_{s\in N(r)\setminus A} f(l_s^{r-1}) \cdot U_{A\cup \{s\}}^{r-1} + \sum_{s\in A\cap N(r)} f(l_s^{r-1}) \cdot U_{A}^{r-1}} \right],
    \end{align*}
    due to the fact that $U_{s}^{r-1} \in \{0,1\}$ for any $s\in S$.
    Notice that $g(x,y) = \frac{xy}{x+y}$ is concave, by Jensen's Inequality, we have
    \begin{align*}
        \Pr[U_{A}^{r}=1] &\leq \frac{\E_{r-1}[U_{A}^{r-1}] \cdot \sum_{s\in N(r)\setminus A} f(l_s^{r-1}) \cdot \E_{r-1}\left[U_{A\cup \{s\}}^{r-1}\right]}{\sum_{s\in N(r)\setminus A} f(l_s^{r-1}) \cdot \E_{r-1}\left[U_{A\cup \{s\}}^{r-1}\right] + \sum_{s\in A\cap N(r)} f(l_s^{r-1}) \cdot \E_{r-1}\left[U_{A}^{r-1}\right]} \\
        &= \frac{\Pr[U_{A}^{r-1}=1] \cdot \sum_{s\in N(r)\setminus A} f(l_s^{r-1}) \cdot \Pr\left[U_{A\cup \{s\}}^{r-1}=1\right]}{\sum_{s\in N(r)\setminus A} f(l_s^{r-1}) \cdot \Pr\left[U_{A\cup \{s\}}^{r-1}=1\right] + \sum_{s\in A\cap N(r)} f(l_s^{r-1}) \cdot \Pr\left[U_{A}^{r-1}=1\right]}.
    \end{align*}

    Now we prove the lemma by mathematical induction on $r$.
    
    When $r=0$, the statement holds trivially since $f(0)=1$ for any $s\in A$.
    Next, we assume that the statement holds for $r-1$.
    By induction hypothesis and monotonicity of $g(x,y) = \frac{xy}{x+y}$, we have
    \begin{equation*}
        \Pr[U_{A}^{r}=1] \leq \frac{1}{\prod_{s\in A} f(l_s^{r-1})} \cdot \frac{|N(r)\setminus A|}{|N(r)\setminus A| + \sum_{s\in A\cap N(r)}  f(l_s^{r-1})}.
    \end{equation*}
    Suppose that $|A\cap N(r)| = m \leq d(r)$, we can rewrite the second item of RHS as
    \begin{equation*}
        \frac{|N(r)\setminus A|}{|N(r)\setminus A| + \sum_{s\in A\cap N(r)} \cdot f(l_s^{r-1})} = \frac{d(r)-m}{d(r)-m+\sum_{s\in A\cap N(r)}f(l_s^{r-1})}.
    \end{equation*}
    Since $d(r) \leq d$, we have
    \begin{equation*}
        \frac{|N(r)\setminus A|}{|N(r)\setminus A| + \sum_{s\in A\cap N(r)} f(l_s^{r-1})} \leq \frac{d-m}{d-m+\sum_{s\in A\cap N(r)}f(l_s^{r-1})} \leq \prod_{s\in A\cap N(r)} \frac{f(l_s^{r-1})}{f(l_s^{r-1}+1)},
    \end{equation*}
    where the last inequality holds due to Equation~\eqref{equation:property-of-f}.
    Notice that $f(l_s^r) = f(l_s^{r-1}+1)$ if and only if $s\in N(r)$, otherwise $f(l_s^r) = f(l_s^{r-1})$, therefore it holds that
    \begin{equation*}
        \Pr[U_{A}^{r}=1] \leq \frac{1}{\prod_{s\in A} f(l_s^{r-1})} \cdot \prod_{s\in A} \frac{f(l_s^{r-1})}{f(l_s^r)} = 1/\prod_{s\in A} f(l_s^r),
    \end{equation*}
    and we finish the proof.
\end{proofof}

\begin{proofof}{Claim~\ref{claim:g_10000(10000)}}
    We define the following function, where $x\geq 1$ and $t\in (0,1)$:
    \begin{equation*}
        \alpha(x,t) := \left( \frac{1+xt}{1-t} \right)^\frac{1}{t}.
    \end{equation*}
    
    By definition, we have
    \begin{equation*}
        g_{d}(l)=g_{d}(l-1)\cdot \left( \min_{t\in (0,1)} \left\{ \alpha(g_{d}(l-1),t) \right\} \right)^{\frac{1}{d}}.
    \end{equation*}

    Therefore, in the following, we focus on lower bounding $\min_{t\in (0,1)} \{\alpha(x,t)\}$ for any given $x\in [1,10)$ (recall that $g_d$ is an increasing function with $g_d(0)=1$, and our goal is to show that $g_d(d) \geq 9.7657$, where $d=10000$).
    We first split the interval $(0,1)$ into three parts: $(0,10^{-6})$, $[10^{-6},1-10^{-6}]$, and $(1-10^{-6},1)$, and then provide a lower bound on $\alpha(x,t)$ for each of these intervals.
    When $t \in (0,10^{-6})$, we have
    \begin{equation*}
        \alpha(x, t) = \left( \frac{1+xt}{1-t} \right)^{\frac{1}{t}} \geq (1+xt)^{\frac{1}{t}} \geq (1+10^{-6}\cdot x)^{10^6},
    \end{equation*}
     where the second inequality holds since $\log((1+xt)^{1/t})$ is decreasing in $t$:
     its derivative $(\frac{xt}{1+xt}-\log(1+xt))/t^2$ is negative because the numerator is decreasing in $t$ and equals $0$ when $t=0$.
          
    %
    
    When $t\in (1-10^{-6},1)$, we have 
    \begin{equation*}
        \alpha(x,t) = \left( \frac{1+xt}{1-t} \right)^{\frac{1}{t}} \geq \frac{1+xt}{1-t} \geq 10^{6} \geq (1+10^{-6}\cdot x)^{10^6}.
    \end{equation*}

    Next, we consider the case when $t\in [10^{-6},1-10^{-6}]$.
    We show that there exists only one local minimum of $\alpha(x,t)$ by showing the log-convexity of $\alpha(x,t)$ (and therefore $\alpha(x,t)$ is convex in $t$).
    
    Taking the second partial derivative of $\log(\alpha(x,t))$ on $t$, we have
    \begin{equation*}
        \frac{\partial^2 \log(\alpha(x,t))}{\partial t^2} = \frac{(4xt^2-3(x-1)t-2)\cdot (x+1)\cdot t + 2(1-t)^2\cdot (xt+1)^2\cdot \log\left(\frac{xt+1}{1-t}\right)}{(t-1)^2\cdot (xt+1)^2\cdot t^3}.
    \end{equation*}
    
    Next we let $\mu = 1-t\in (0,1)$, $\rho = 1+xt \in (1, 11)$, and rewrite the numerator of the RHS as
    \begin{equation*}
        (-4\mu\rho + \mu + \rho) (\rho-\mu) + 2\mu^2\rho^2 \log (\rho/\mu) := \psi(\mu,\rho).
    \end{equation*}

    Taking partial derivative and second partial derivative of $\psi(\mu,\rho)$ over $\rho$, we have
    \begin{align*}
        \frac{\partial \psi (\mu,\rho)}{\partial \rho} = 2(\mu^2(\rho+2) + 2\mu^2\rho \log (\rho/\mu) - 4\mu\rho + \rho),
    \end{align*}
    and
    \begin{align*}
       \frac{\partial^2 \psi(\mu,\rho)}{\partial \rho^2} = 4\mu^2\log (\rho/\mu) + 6\mu^2-8\mu+2 
       \geq \ 4\mu^2\log \left( \frac{1}{\mu} \right) + 6\mu^2 - 8\mu + 2 := \kappa(\mu).
    \end{align*}

    Notice that $\kappa'(\mu) = 8 (\mu+\mu\log (1/\mu) - 1)$ and
    $\kappa''(\mu) = 8\log (1/\mu) > 0$ since $\mu\in(0,1)$.
    Therefore, we have $\kappa'(\mu)  < \kappa'(1) = 0$.
    Hence, $\kappa(\mu)$ is decreasing on $\mu \in (0,1)$.
    We have $\kappa(\mu) > \kappa(1) = 0$, which implies
    \begin{align*}
        \frac{\partial^2 \psi(\mu,\rho)}{\partial \rho^2} \geq \frac{\partial^2 \psi(\mu,\rho)}{\partial \rho^2} \bigg|_{\rho=1} = 0.
    \end{align*}
    Hence $\frac{\partial \psi(\mu,\rho)}{\partial \rho}$ is increasing on $\rho \in (1,11)$, and we have
    \begin{align*}
        \frac{\partial \psi(\mu,\rho)}{\partial \rho} \geq \frac{\partial \psi(\mu,\rho)}{\partial \rho}\bigg|_{\rho=1} = 0,
    \end{align*}
    which indicates that $\psi(\mu,\rho)$ is increasing on $\rho \in (1,11)$.
    Therefore, we have
    \begin{align*}
        \psi(\mu,\rho) \geq \psi(\mu,1) = (3\mu-1)(\mu-1) +2\mu^2 \log (1/\mu) \geq 0,
    \end{align*}
    which yields the log-convexity of $\alpha(x,t)$.
    
    
    The above argument enables us to use binary search to find where $\argmin_t \alpha(x,t)$ is located to any precision.
    We use the right side point as the result of our binary search and let $\delta(x)$ be the minimum of $\alpha(x,t)$ in $t\in [10^{-6},1-10^{-6}]$ we compute using binary search with a precision of $10^{-15}$.
    In the following, we show that the error between $\delta(x)$ and the true minimum is at most $10^{-6}$.
    We first upper bound the derivative of $\alpha(x,t)$ on $t$ as follows
    \begin{align*}
        \frac{\partial \alpha(x,t)}{\partial t} &\leq \frac{\partial \alpha(x,t)}{\partial t}\vert_{t=1-10^{-6}} = \frac{\left( \frac{xt+1}{1-t}\right)^{\frac{1}{t}}\cdot \left((1-t)(xt+1)\log \left( \frac{xt+1}{1-t} \right) -(x+1)t\right)}{(1-t)t^2(xt+1)}\vert_{t=1-10^{-6}} \\
        &\leq \frac{\left( \frac{xt+1}{1-t}\right)^{\frac{1}{t}}\cdot (1-t)(xt+1)\log \left( \frac{xt+1}{1-t} \right)}{(1-t)t^2(xt+1)}\vert_{t=1-10^{-6}} = \frac{\left( \frac{xt+1}{1-t}\right)^{\frac{1}{t}}\cdot \log \left( \frac{xt+1}{1-t} \right)}{t^2}\vert_{t=1-10^{-6}} \leq 10^{9},
    \end{align*}
    where the first inequality holds due to the convexity of $\alpha(x,t)$ while the last holds since $x\leq 10$.
    Hence when the binary searching stops, we can guarantee that the true minimum in the final searching interval (with size $\leq 10^{-15}$) is at least $\delta(x) - 10^{-6}$. 

    In conclusion, for any $x\in [1,10)$ we have
    \begin{equation*}
        \min_{t\in (0,1)} \{\alpha(x,t)\} \geq \min\{(1+10^{-6}\cdot x)^{10^6},\delta(x)-10^{-6}\}.
    \end{equation*}
    
    By numerical computation, it can be verified that $g_{10000}(10000)$ is at least $9.7657$.
\end{proofof}

\section{Hardness for Randomized Algorithms} \label{sec:problem-hardness}

In this section, we show hardness results for randomized algorithms for online bipartite matching on $d$-regular graphs (and thus also for $(d,d)$-bounded graphs).
Recall that Cohen and Wajc~\cite{conf/soda/CohenW18} showed that no randomized algorithm can achieve a competitive ratio better than $1-1/ \sqrt{8\pi d} = 1 - O(1/ \sqrt{d})$ when $d\to \infty$.
In the following, we revisit the hard instance of~\cite{conf/soda/CohenW18} and show that the same (class of) hard instance can be used to show an upper bound for all (randomized) algorithms for every $d\geq 2$.

\begin{theorem}\label{theorem:problem-hardness}
For the online bipartite matching problem on $d$-regular graphs, where $d\geq 2$, no randomized algorithm can achieve a competitive ratio better than $\eta(d)$, where 
    \begin{equation*}
        \eta(d) = 1 - \sum_{i=\lceil d/2 \rceil + 1}^{d} \left( \frac{i-\lceil d/2 \rceil}{d} \binom{d}{i} \left(\frac{\lfloor d/2 \rfloor}{d}\right)^{d-i} \left(1-\frac{\lfloor d/2 \rfloor}{d}\right)^i \right).
    \end{equation*}
\end{theorem}
\begin{proof}
    Observe that for the online bipartite matching problem, it suffices to consider greedy algorithms that never leave a request unmatched when it has unmatched neighbors.
    Also note that every $d$-regular graph admits a perfect matching (by Hall's Theorem).
    Thus it suffices to provide upper bounds on the expected number of matched servers for the online algorithm.
    
    Consider an instance with $d^2$ servers, in which the servers are indexed by $(i,j)$, for $i,j\in \{1,2,\ldots,d\}$.
    We use $s_{i,j}$ to denote the server with index $(i,j)$ (see Figure~\ref{fig:hard-instance} as an example).
    We organize the servers into groups of size $d$.
    For each $i\in [d]$, let $T_i = \{s_{i,1},..., s_{i,d}\}$ be the $i$-th horizontal group.
    For each $j\in [d]$, let $T'_j = \{s_{1,j},..., s_{d,j}\}$ be the $j$-th vertical group.
    In the following, we introduce the online requests with two phases of arrivals.
        
    In Phase 1, for each horizontal group $T_i$, there are $\lfloor d/2 \rfloor$ online requests $r_{i,1},..., r_{i,\lfloor d/2 \rfloor}$ that are neighbors of all the servers in $T_i$.
    Note that after Phase 1, each horizontal group has exactly $\lfloor d/2 \rfloor$ matched servers.
    Since the requests are indistinguishable, by assigning the index uniformly at random, we can ensure that for all (randomized) algorithms, each server is matched with probability $\lfloor d/2 \rfloor/d$ after Phase 1.
    Moreover, for the servers in the same vertical group, whether they are matched are independent random events.
    Hence the probability $p_l$ that there are exactly $l$ unmatched servers in each vertical group is
    \begin{equation*}
        p_l = \binom{d}{l} \left(\frac{\lfloor d/2 \rfloor}{d}\right)^{d-l} \left(1-\frac{\lfloor d/2 \rfloor}{d}\right)^l,
    \end{equation*}
    where $l \in \{0,1, \ldots, d\}$.
    
    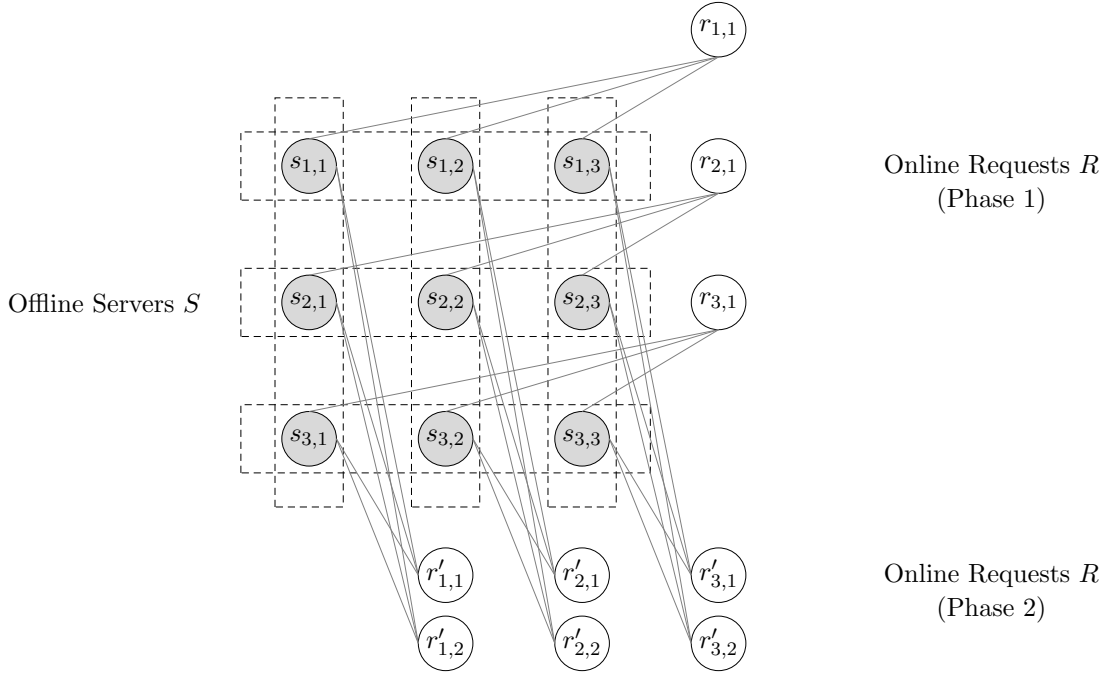
\begin{figure}[htb]
        \begin{center}
        \resizebox{.9\textwidth}{!}{%
        \begin{tikzpicture}
                \node at (-3,2) {Offline Servers $S$};
                \draw [fill = gray!30](0,0) circle(0.4);    \node at (0,0) {$s_{3,1}$};
                \draw [fill = gray!30](2,0) circle(0.4);    \node at (2,0) {$s_{3,2}$};
                \draw [fill = gray!30](4,0) circle(0.4);    \node at (4,0) {$s_{3,3}$};
                \draw [densely dashed] (-1,-0.5) rectangle (5,0.5);
                \draw [densely dashed] (-0.5,-1) rectangle (0.5,5);
                \draw [fill = gray!30](0,2) circle(0.4);    \node at (0,2) {$s_{2,1}$};
                \draw [fill = gray!30](2,2) circle(0.4);    \node at (2,2) {$s_{2,2}$};
                \draw [fill = gray!30](4,2) circle(0.4);    \node at (4,2) {$s_{2,3}$};
                \draw [densely dashed] (-1,1.5) rectangle (5,2.5);
                \draw [densely dashed] (1.5,-1) rectangle (2.5,5);
                \draw [fill = gray!30](0,4) circle(0.4);    \node at (0,4) {$s_{1,1}$};
                \draw [fill = gray!30](2,4) circle(0.4);    \node at (2,4) {$s_{1,2}$};
                \draw [fill = gray!30](4,4) circle(0.4);    \node at (4,4) {$s_{1,3}$};
                \draw [densely dashed] (-1,3.5) rectangle (5,4.5);
                \draw [densely dashed] (3.5,-1) rectangle (4.5,5);
                \node at (10,4) {Online Requests $R$};
                \node at (10,3.5) {(Phase $1$)};
                \draw (6,2) circle(0.4);    \node at (6,2) {$r_{3,1}$};
                \draw [gray](0,0.4)--(6,1.6); \draw [gray](2,0.4)--(6,1.6); \draw [gray](4,0.4)--(6,1.6);
                %
                \draw (6,4) circle(0.4);    \node at (6,4) {$r_{2,1}$};
                \draw [gray](0,2.4)--(6,3.6); \draw [gray](2,2.4)--(6,3.6); \draw [gray](4,2.4)--(6,3.6);
                %
                \draw (6,6) circle(0.4);    \node at (6,6) {$r_{1,1}$};
                \draw [gray](0,4.4)--(6,5.6); \draw [gray](2,4.4)--(6,5.6); \draw [gray](4,4.4)--(6,5.6);
                %
                \node at (10,-2) {Online Requests $R$};
                \node at (10,-2.5) {(Phase $2$)};
                \draw  (2,-2) circle(0.4);    \node at (2,-2) {$r'_{1,1}$};
                \draw [gray](0.4,0)--(1.6,-2); \draw [gray](0.4,2)--(1.6,-2); \draw [gray](0.4,4)--(1.6,-2);
                \draw (2,-3) circle(0.4);    \node at (2,-3) {$r'_{1,2}$};
                \draw [gray](0.4,0)--(1.6,-3); \draw [gray](0.4,2)--(1.6,-3); \draw [gray](0.4,4)--(1.6,-3);
                \draw (4,-2) circle(0.4);    \node at (4,-2) {$r'_{2,1}$};
                \draw [gray](2.4,0)--(3.6,-2); \draw [gray](2.4,2)--(3.6,-2); \draw [gray](2.4,4)--(3.6,-2);
                \draw (4,-3) circle(0.4);    \node at (4,-3) {$r'_{2,2}$};
                \draw [gray](2.4,0)--(3.6,-3); \draw [gray](2.4,2)--(3.6,-3); \draw [gray](2.4,4)--(3.6,-3);
                \draw (6,-2) circle(0.4);    \node at (6,-2) {$r'_{3,1}$};
                \draw [gray](4.4,0)--(5.6,-2); \draw [gray](4.4,2)--(5.6,-2); \draw [gray](4.4,4)--(5.6,-2);
                \draw (6,-3) circle(0.4);    \node at (6,-3) {$r'_{3,2}$};
                \draw [gray](4.4,0)--(5.6,-3); \draw [gray](4.4,2)--(5.6,-3); \draw [gray](4.4,4)--(5.6,-3);
            \end{tikzpicture}
            }
            \caption{Hard instance for randomized algorithms: an illustrating example when $d=3$. During Phase 1 each horizontal group $T_i$ has an online requests $r_{i,1}$ that is neighbor to all three servers $s_{i,1}, s_{i,2}, s_{i,3}$. During Phase 2 each vertical group $T'_j$ has $2$ online requests $r'_{j,1}, r'_{j,2}$ that are neighbors to all three servers $s_{1,j}, s_{2,j}, s_{3,j}$.}
        \label{fig:hard-instance}
        \end{center}
    \end{figure}
    
    In Phase 2, for each vertical group $T'_j$, let there be $\lceil d/2 \rceil$ requests that are common neighbors of all the servers in $T'_j$.
    Note that if there are $l$ unmatched servers in a vertical group when Phase 2 begins, where $l > \lceil d/2 \rceil$, then at the end of the algorithm, $l - \lceil d/2 \rceil$ servers are unmatched in this vertical group.
    Hence the expected number of unmatched servers in each vertical group is
    \begin{equation*}
        \sum_{l = \lceil d/2 \rceil+1}^d \left( (l-\lceil d/2 \rceil)\cdot p_l \right).
    \end{equation*}
    
    Therefore, the competitive ratio of a randomized algorithm is upper-bounded by
    \begin{equation*}
        \frac{1}{d} \left(d - \sum_{l = \lceil d/2 \rceil+1}^d \left( (l-\lceil d/2 \rceil)\cdot p_l \right) \right) = 1 - \sum_{l = \lceil d/2 \rceil + 1}^{d} \left( \frac{l - \lceil d/2 \rceil}{d} \binom{d}{l} \left(\frac{\lfloor d/2 \rfloor}{d}\right)^{d-l} \left(1-\frac{\lfloor d/2 \rfloor}{d}\right)^l \right),
    \end{equation*}
    and we finish the proof.
\end{proof}

Cohen and Wajc~\cite{conf/soda/CohenW18} showed that $\eta(d)$ approaches $1-\frac{1}{\sqrt{8\pi d}}$ when $d\to \infty$ using the normal approximation of the binomial distribution.
For a more illustrative presentation for bounded $d$, we list the values of $\eta(d)$ for $d\in \{2,\ldots,10\}$ as follows.

\begin{table}[htb]
    \begin{center}
        \begin{tabular}{ | c | c|c|c|c|c|c|c|c|c| } 
            \hline
            $d$ & $2$ & $3$ & $4$ & $5$ & $6$ & $7$ & $8$ & $9$ & $10$ \\
            \hline 
            $\eta(d)$ & $0.875$ & $0.9013$ & $0.9063$ & $0.9171$ & $0.9219$ & $0.9281$ & $0.9317$ & $0.9358$ & $0.9385$ \\
            \hline
            \end{tabular}
        \end{center}
    \caption{Upper bounds on the competitive ratio of randomized algorithms with $d\in \{2,\ldots,10\}$.}
\end{table}

\section{Extensions to \texorpdfstring{$(k,d)$}{}-bounded Graphs}
\label{appendix:kd}

In this section, we generalize the upper bound for \Ranking and lower bound for OCS from $(d,d)$-bounded graphs to $(k,d)$-bounded graphs, with $k > d\geq 2$.

\subsection{Hardness for Ranking}

\paragraph{Hard Instance for General $k$ and $d$.}
Our construction can be viewed as a generalization of that in Section~\ref{sec:hardness-for-Ranking-general-d} (see Figure~\ref{fig:hardness-for-Ranking-(4,3)} as an example when $k=4$ and $d=3$).
Let $S = S_1 \cup S_2$, where $S_1 = \{s_1,\ldots,s_k\}$ and $S_2 = \{s_{k+1},\ldots,s_{k+d-1}\}$.
Let $R_1 = \{r_1,\ldots,r_k\}$ and $R_2 = \{r_{k+1},\ldots,r_{k+\lceil k(k-1)/d \rceil}\}$.
For $i \in \{1,\dots,k\}$, let there be an edge between $s_i$ and $r_i$.
Then we add edges such that $R_1$ and $S_2$ form a complete bipartite graph.
Finally, for each server $s_i$ in $S_1$, we pick arbitrary $k-1$ requests in $R_2$ with minimum degree as its neighbor (this is possible because $|R_2| = \left\lceil \frac{k(k-1)}{d} \right\rceil \geq \frac{k-1}{d}\cdot |S_1|$).
It can be shown that all servers have degree $k$ and all requests have degree at most $d$.
Let requests arrive in the order of $r_1,r_2,\ldots$.

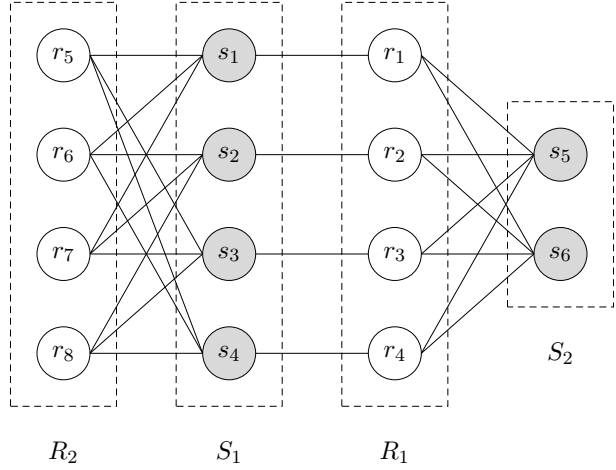
\begin{figure}[htb]
\begin{center}
\resizebox{0.5\textwidth}{!}{
\begin{tikzpicture}
\draw [fill = gray!30] (2.5,0) circle (0.4); \node at (2.5,0) {$s_4$};
\draw [fill = gray!30] (2.5,1.5) circle (0.4); \node at (2.5,1.5) {$s_3$};
\draw [fill = gray!30] (2.5,3) circle (0.4); \node at (2.5,3) {$s_2$};
\draw [fill = gray!30] (2.5,4.5) circle (0.4); \node at (2.5,4.5) {$s_1$};
\draw [fill = gray!30] (7.5,1.5) circle (0.4); \node at (7.5,1.5) {$s_6$};
\draw [fill = gray!30] (7.5,3) circle (0.4); \node at (7.5,3) {$s_5$};
\draw (0,0) circle (0.4); \node at (0,0) {$r_8$};
\draw (0,1.5) circle (0.4); \node at (0,1.5) {$r_7$};
\draw (0,3) circle (0.4); \node at (0,3) {$r_6$};
\draw (0,4.5) circle (0.4); \node at (0,4.5) {$r_5$};
\draw (5,0) circle (0.4); \node at (5,0) {$r_4$};
\draw (5,1.5) circle (0.4); \node at (5,1.5) {$r_3$};
\draw (5,3) circle (0.4); \node at (5,3) {$r_2$};
\draw (5,4.5) circle (0.4); \node at (5,4.5) {$r_1$};
\draw (2.9,4.5)--(4.6,4.5); 
\draw (2.9,3)--(4.6,3); 
\draw (2.9,1.5)--(4.6,1.5); 
\draw (2.9,0)--(4.6,0); 
\draw (5.4,4.5)--(7.1,3); 
\draw (5.4,3)--(7.1,3); 
\draw (5.4,1.5)--(7.1,3); 
\draw (5.4,0)--(7.1,3); 
\draw (5.4,4.5)--(7.1,1.5); 
\draw (5.4,3)--(7.1,1.5); 
\draw (5.4,1.5)--(7.1,1.5); 
\draw (5.4,0)--(7.1,1.5); 
\draw (0.4,4.5)--(2.1,4.5); 
\draw (0.4,3)--(2.1,4.5); 
\draw (0.4,1.5)--(2.1,4.5); 
\draw (0.4,3)--(2.1,3); 
\draw (0.4,1.5)--(2.1,3); 
\draw (0.4,0)--(2.1,3); 
\draw (0.4,1.5)--(2.1,1.5); 
\draw (0.4,0)--(2.1,1.5); 
\draw (0.4,4.5)--(2.1,1.5); 
\draw (0.4,0)--(2.1,0); 
\draw (0.4,4.5)--(2.1,0); 
\draw (0.4,3)--(2.1,0); 
\draw [densely dashed] (-0.8,-0.8) rectangle (0.8,5.3); \node at (0,-1.5) {$R_2$};
\draw [densely dashed] (1.7,-0.8) rectangle (3.3,5.3); \node at (2.5,-1.5) {$S_1$};
\draw [densely dashed] (4.2,-0.8) rectangle (5.8,5.3); \node at (5,-1.5) {$R_1$};
\draw [densely dashed] (6.7,0.7) rectangle (8.3,3.8); \node at (7.5,0) {$S_2$};
\end{tikzpicture} }
\end{center}
\vspace{-10pt}
\caption{Hard instance for \Ranking: an illustrating example when $k=4$ and $d=3$.}
\label{fig:hardness-for-Ranking-(4,3)}
\end{figure}

\begin{theorem}
    The competitive ratio of \textnormal{\Ranking} for online bipartite matching on $(k,d)$-bounded graphs is at most $1-\left((1-1/d)^k\cdot (d-1)\right)/(k+d-1)$.
\end{theorem}
\begin{proof}
    (\textit{sketch})
    Following the equivalent description of \Ranking in Section~\ref{sec:hardness-for-Ranking-general-d}, we first realize the ranks of the servers in $S_2$ and let the ordered ranks be $0\leq \theta_1\leq \cdots\leq \theta_{d-1}$, then realize the rank $y_i$ of $s_i$ upon the arrival of $r_i$ for $i\in \{1,\ldots,k\}$.
    Following the same analysis as in Section~\ref{sec:hardness-for-Ranking-general-d}, we can use Jensen's Inequality to show that it suffices to consider the case when $(\theta_1,\ldots,\theta_{d-1}) = (\frac{1}{d},\ldots,\frac{d-1}{d})$.
    Let $\beta(i)$ be a random variable denoting the number of unmatched servers in $S_2$ after $r_i$ arrives, we can show that $\E_{y_i}[\beta(i) \mid \beta(i-1)] = \frac{d-1}{d}\cdot \beta(i-1)$.
    Thus the competitive ratio of \Ranking in our instance is upper bounded by
    \begin{equation*}
        \frac{k+d-1-\E[\beta (k)]}{k+d-1} \leq 1-\frac{\left(1-\frac{1}{d}\right)^k\cdot (d-1)}{k+d-1}.
    \qedhere
    \end{equation*}
\end{proof}

\begin{corollary}
    Let $\alpha = k/d > 1$.
    When $d\to \infty$, the competitive ratio of \textnormal{\Ranking} on $(k,d)$-bounded graphs is at most $(1 - \frac{e^{-\alpha}}{\alpha+1})$.
\end{corollary}
\begin{proof}
    When $d \to \infty$, the competitive ratio of \Ranking is upper bounded by
    \begin{equation*}
        \lim_{d\to \infty} \left\{ 1-\frac{\left(1-\frac{1}{d}\right)^k\cdot (d-1)}{k+d-1} \right\} = \lim_{d\to \infty} \left\{ 1 - \frac{d-1}{(\alpha+1)d-1} \cdot \left( 1-\frac{1}{d} \right)^{\alpha\cdot d} \right\} = 1 - \frac{e^{-\alpha}}{\alpha+1}.
        \qedhere
    \end{equation*}
\end{proof}

\subsection{OCS-based Algorithm}

Recall that in Algorithm~\ref{alg:discrete-ocs}, we fix the optimal candidate function and use it to decide the probability of matching requests to unmatched servers.
Also recall that in Section~\ref{ssec:characterize-candidate}, the optimal candidate function is defined for all integers (including integers larger than $d$).
Therefore, for $(k,d)$-bounded graphs, we can use the optimal candidate function in $\mathcal{C}_d$, and bound the probability that a server $s\in S$ is unmatched by $1/f(k)$.

\begin{theorem}
    Algorithm~\ref{alg:discrete-ocs} (with candidate function $f: \{0,1,\ldots,k\} \to \mathbb{R}$) achieves a competitive ratio of $1-1/f(k)$ for the (vertex-weighted) online bipartite matching on $(k,d)$-bounded graphs.
\end{theorem}

We use the optimal candidate function $f^*_d\in \mathcal{C}_d$ defined in Definition~\ref{definition:optimal-candidate} and analyze the competitive ratio of Algorithm~\ref{alg:discrete-ocs} when $d$ turns to infinity.

\begin{theorem} \label{theorem:OCS-(k,d)}
    Let $\alpha = k/d > 1$.
    When $d\to \infty$, Algorithm~\ref{alg:discrete-ocs} achieves a competitive ratio at least $1-(\frac{1}{3e})^\alpha$ for $(k,d)$-bounded graphs.
\end{theorem}
\begin{proof}
    We show that $f^*_d(\alpha\cdot d) \geq (3e)^\alpha$.
    By definition, $f^*_d$ is an increasing function.
    Then for any $l\geq 1$, we have $f^*_d(l+1) \geq f^*_d(l)$, which implies that 
    \begin{equation*}
        \min_{1\leq m\leq d} \left\{ \left(1 + \frac{m\cdot f^*_d(l+1)}{d-m}\right)^\frac{1}{m} \right\} \geq \min_{1\leq m\leq d} \left\{\left(1 + \frac{m\cdot f^*_d(l)}{d-m}\right)^\frac{1}{m} \right\}.
    \end{equation*}
    
    Hence we have $f^*_d(l+1)/f^*_d(l) \geq f^*_d(l)/f^*_d(l-1)$, leading to
    \begin{align*}
        f^*_d(k) &= f^*_d(d) \cdot \prod_{i=0}^{k-d-1} \frac{f^*_d(d+i+1)}{f^*_d(d+i)} \geq f^*_d(d) \cdot \left( \frac{f^*_d(d)}{f^*_d(d-1)} \right)^{k-d} \\
        &\geq f^*_d(d) \cdot \left( f^*_d(d) \right)^{\frac{k-d}{d}} = f^*_d(d)^\alpha \geq (3e)^\alpha,
    \end{align*}
    where the second inequality holds due to
    \begin{equation*}
        f^*_d(d) = \prod_{i=0}^{d-1} \frac{f^*_d(i+1)}{f^*_d(i)} \leq \left(\frac{f^*_d(d)}{f^*_d(d-1)} \right)^d,
    \end{equation*}
    and the last inequality holds since $\lim_{d\to\infty} f^*_d(d) \geq 9.7 \geq 3e$ (by Lemma~\ref{lemma:lower-bound-f-by-g} and Claim~\ref{claim:g_10000(10000)}).
\end{proof}

\begin{corollary}
    When $k > d \to \infty$, the competitive ratio of the OCS-based algorithm (Algorithm~\ref{alg:discrete-ocs}) is larger than that of \textnormal{\Ranking}.
\end{corollary}
\begin{proof}
    It suffices to show that $1-(\frac{1}{3e})^\alpha > 1 - \frac{e^{-\alpha}}{\alpha+1}$ for all $\alpha \geq 1$.
    Note that $3^\alpha > \alpha + 1$ holds for all $\alpha \geq 1$.
    Hence we have
    \begin{equation*}
        (3e)^\alpha > (\alpha+1) \cdot e^\alpha 
        \ \Rightarrow \
        \left( \frac{1}{3e} \right)^\alpha < \frac{1}{(\alpha+1) \cdot e^\alpha} = \frac{e^{-\alpha}}{\alpha+1}
        \ \Rightarrow \
        1-\left(\frac{1}{3e}\right)^\alpha > 1 - \frac{e^{-\alpha}}{\alpha+1},
    \end{equation*}
    which completes the proof.
\end{proof}

Finally, we provide the numerical comparison between the upper bound for \Ranking and the lower bound for OCS on $(k,d)$-bounded graphs when $k$ and $d$ are small (see Table~\ref{tab:U.B.-of-Ranking(k,d)}), showing that the ratio of OCS is consistently larger than that of \Ranking.

\begin{table}[htb]
    \begin{center}
        \begin{tabular}{ c || c|c|c|c|c|c } 
            $k \backslash d$ & $2$ & $3$ & $4$ & $5$ & $6$ & $7$ \\
            \hline 
            \hline
            $3$ & $\mathbf{0.992}/0.968$ & - & - & - & - & - \\ 
            $4$ & $\mathbf{0.999}/0.987$ &$\mathbf{0.954}/0.934$ & - & - & - & - \\ 
            $5$ & $\mathbf{0.999}/0.994$ &$\mathbf{0.993}/0.962$ & $\mathbf{0.943}/0.911$ & - & - & - \\ 
            $6$ & $\mathbf{0.999}/0.997$ &$\mathbf{0.999}/0.978$ & $\mathbf{0.985}/0.941$ & $\mathbf{0.933}/0.895$ & - & - \\
            $7$ & $\mathbf{0.999}/0.999$ &$\mathbf{0.999}/0.987$ & $\mathbf{0.997}/0.960$ & $\mathbf{0.976}/0.924$ & $\mathbf{0.928}/0.884$ & - \\
            $8$ & $\mathbf{0.999}/0.999$ &$\mathbf{0.999}/0.992$ & $\mathbf{0.999}/0.973$ & $\mathbf{0.993}/0.944$ & $\mathbf{0.968}/0.911$ & $\mathbf{0.924}/0.875$ \\
            $9$ & $\mathbf{0.999}/0.999$ &$\mathbf{0.999}/0.995$ & $\mathbf{0.999}/0.981$ & $\mathbf{0.998}/0.959$ & $\mathbf{0.988}/0.931$ & $\mathbf{0.962}/0.900$ \\
            $10$ & $\mathbf{0.999}/0.999$ &$\mathbf{0.999}/0.997$ & $\mathbf{0.999}/0.987$ & $\mathbf{0.999}/0.969$ & $\mathbf{0.996}/0.946$ & $\mathbf{0.983}/0.920$ \\
            \end{tabular}
        \end{center}
    \caption{The comparison of lower bounds of OCS and upper bounds of \Ranking for small $k$ and $d$. For each cell with $k>d$, we use $\cdot / \cdot$ to present the lower bound for OCS and the upper bound for \Ranking. It can be observed that for each presented $k$ and $d$, OCS outperforms \Ranking (for which we use the \textbf{bold} numbers to indicate).
    For $d=2$, the ratio is obtained by the optimal two-way semi-OCS in \cite{conf/focs/GaoHHNYZ21}; for $d\geq 3$ it is obtained by the optimal candidate function in $\mathcal{C}_d$.}
    \label{tab:U.B.-of-Ranking(k,d)}
\end{table}

\end{document}